\documentclass[twoside,11pt]{article}
\pdfoutput=1
\usepackage{jmlr2e}
\usepackage{amsmath}
\usepackage{bbm}
\usepackage{mathtools}
\usepackage{xcolor}
\setcitestyle{authoryear,open={},aysep={},close={}}

\DeclareMathOperator*{\argmin}{arg\hspace{-.07cm}\,min}
\newcommand{\polylog}{{\hspace{.05cm} \mbox{polylog}}}

\jmlrheading{1}{2022}{1-48}{4/00}{10/00}{OscarLopez}{Oscar L\'opez}

\ShortHeadings{Near-Optimal Weighted Matrix Completion}{Oscar L\'opez}
\firstpageno{1}

\begin{document}

\title{Near-Optimal Weighted Matrix Completion}

\author{\name Oscar L\'opez \email lopezo@fau.edu \\
       \addr Harbor Branch Oceanographic Institute\\
       Florida Atlantic University\\
       Fort Pierce, FL 34946, USA
       }

\editor{}

 \maketitle

\begin{abstract}
\textcolor{black}{Recent work in the matrix completion literature has shown that prior knowledge of a matrix's row and column spaces can be successfully incorporated into reconstruction programs to substantially benefit matrix recovery. This paper proposes a novel methodology that exploits more general forms of known matrix structure in terms of subspaces. The work derives reconstruction error bounds that are informative in practice, providing insight to previous approaches in the literature while introducing novel programs that severely reduce sampling complexity. The main result shows that a family of weighted nuclear norm minimization programs incorporating a $M_1 r$-dimensional subspace of $n\times n$ matrices (where $M_1\geq 1$ conveys structural properties of the subspace) allow accurate approximation of a rank $r$ matrix aligned with the subspace from a near-optimal number of observed entries (within a logarithmic factor of $M_1 r)$. The result is robust, where the error is proportional to measurement noise, applies to full rank matrices, and reflects degraded output when erroneous prior information is utilized. Numerical experiments are presented that validate the theoretical behavior derived for several example weighted programs.}
\end{abstract}


\section{Introduction}
\label{intro}

Over the past two decades, matrix completion has evolved from an academic curiosity to a common industrial tool (\cite{MC1,MC4,MC7,MC8}). Its utility includes seismic data acquisition (\cite{SMC1,SMC2,SMC3}), machine learning (\cite{ENMC,cluster,classify}), collaborative filtering (\cite{srebroCF}), computer vision (\cite{compvision}), gene expression analysis (\cite{gene}) and MRI (\cite{MRI}). In these applications, practitioners wish to estimate a data matrix of interest $D\in\mathbbm{C}^{n_1\times n_2}$ from a fraction of revealed noisy entries. The success of recent approaches hinges on the underlying assumption that $D$ can be well approximated by a rank $r$ matrix where $r\ll \max\{n_1,n_2\}$, that is, $D$ has \emph{low rank structure}. This data model is common in smooth signals, but pervasive simply by the sheer nature of large scale data (\cite{lowrank}).

To elaborate, let $\Omega\subseteq \{1,\cdots, n_1\}\times\{1,\cdots, n_2\}$ be a subset of size $m\leq n_1n_2$ and $P_{\Omega}:\mathbbm{C}^{n_1\times n_2}\mapsto\mathbbm{C}^{m}$ the corresponding \emph{sampling operator} that extracts the $m$ values at the entries specified by $\Omega$ from an input matrix (with $|\Omega| = m$). Given $P_{\Omega}(D) + d\in\mathbbm{C}^{m}$, where $d\in\mathbbm{C}^{m}$ encompasses measurement noise, the goal of matrix completion is to recover $D$ as accurately as possible. Under the low rank assumption, \textcolor{black}{a well-studied} method to estimate the data matrix is via the \emph{nuclear norm minimization} program (\cite{fazel,MC1,MC4}), \textcolor{black}{which estimates $D$ via}
\begin{equation}
\label{NNmin}
D^{1} \coloneqq \argmin\limits_{X\in\mathbbm{C}^{n_1\times n_2}}\|X\|_* \ \ \mbox{subject to} \ \ \|P_{\Omega}(D) + d - P_{\Omega}(X)\|_2\leq\eta,
\end{equation}
where $\|X\|_* = \sum_{k=1}^{\min(n_1,n_2)}\sigma_k(X)$ is the nuclear norm, $\sigma_k(X)$ is the $k$-th largest singular value of $X$ and $\eta$ is a program parameter chosen according to the noise level. The nuclear norm penalty provides a convex surrogate for the rank objective function, in order to output a low rank matrix that is viable for the noisy observations in a tractable manner.

As in previous work, a notion of \emph{incoherence} is needed to quantify how evenly distributed the information is throughout the data matrix. Ultimately, incoherence conditions ensure that a set of observed entries $\Omega$ chosen uniformly at random provide an appropriate sampling scheme for the array of interest.

\begin{definition} 
\label{coherence}
Given $D\in\mathbbm{C}^{n_1\times n_2}$ and $r\leq\min\{n_1,n_2\}$, consider the singular value decomposition (SVD) $D = U\Sigma V^{\top}$. The $r$-\emph{incoherence parameters} of $D$ are defined as the smallest $\mu_0, \mu_1>0$ such that
\begin{equation}
\label{mu1}
\max_{1\leq k \leq n_1}\sqrt{\sum_{j=1}^{r}\big|U_{kj}\big|^2}\leq \sqrt{\frac{\mu_0r}{n_1}}, \ \ \ \ \ \ \max_{1\leq \ell \leq n_2}\sqrt{\sum_{j=1}^{r}\big|V_{\ell j}\big|^2}\leq \sqrt{\frac{\mu_0r}{n_2}}
\end{equation}
and
\begin{equation}
\label{mu2}
\max_{k,\ell}\sqrt{\sum_{j=1}^{r}\big|U_{kj}\big|^2\big|V_{\ell j}\big|^2} \leq \sqrt{\frac{\mu_1r}{n_1n_2}}.
\end{equation}
\end{definition}
Definition (\ref{mu1}) is common in the literature, know as the \emph{standard incoherence} condition (\cite{optinc}). The parameter $\mu_1$ is unique to this work, but similar to the joint incoherence condition introduced in \cite{MC2}. This novel parameter will be elaborated in Section \ref{incoherence}. Intuitively, small parameters (for example, $\mu_0,\mu_1\sim \log(\max\{n_1,n_2\})$) correspond to data matrices whose information is not concentrated on a few set of entries. Such metrics of ``spikiness'' are necessary when the observations are chosen without regard to the matrix structure, in order to guarantee that the probed entries will supply a substantial amount of information. However, incoherence conditions can be avoided if the sampling scheme is modified according to prior knowledge of the matrix's leverage scores (\cite{ward,ward2,MC_2}). 

The following result states the sampling complexity and resulting error bound for program (\ref{NNmin}), where without loss of generality it is henceforth assumed that $n_1\geq n_2$.

\begin{theorem}
\label{thm1}
Let $D\in\mathbbm{C}^{n_1\times n_2}$ have $r$-incoherence parameters $\mu_0,\mu_1$ and suppose $\Omega\subseteq [n_1]\times [n_2]$ is generated by selecting a subset of size $m\leq n_1n_2$ uniformly at random from all subsets of size $m$. Define $D^{1}$ as in (\ref{NNmin}) with $\|d\|_2\leq \eta$. There exist universal constants $c_0,c_1,c_2>0$ such that if
\[
m\geq c_0\max\{\mu_0,\sqrt{\mu_1}\}n_1r\log^{2}(n_1)
\]
then with high probability
\[
\|D-D^{1}\|_F \leq c_1\sqrt{\frac{n_1n_2\log(n_1)}{m}}\sum_{k=r+1}^{n_2}\sigma_k(D)  \nonumber \\
+ c_2\frac{n_1n_2\sqrt{r}\log(n_1)}{m}\eta.
\]
\end{theorem}
The result matches previous work in terms of sampling complexity required for exact matrix completion (\cite{optinc,MC_2,WMC1}). However, a fair comparison is difficult to make due to the dependence here on $\sqrt{\mu_1}$ whereas other authors only require $\mu_0$. Section \ref{incoherence} will provide data matrices with $\sqrt{\mu_1}<\mu_0$ as well as examples where $\sqrt{\mu_1}\geq\mu_0$ holds. Therefore, there is no strict relationship between these parameters and Theorem \ref{thm1} is arguably on par with incoherence-optimal conditions (\cite{optinc}). These results state that, when a practitioner is oblivious to the data matrix's structure, $n_1r\log^2(n_1)$ observed entries are needed to robustly reconstruct an incoherent $n_1\times n_2$ rank $r$ matrix.

\subsection{\textcolor{black}{Matrix Completion with Prior Knowledge: Approach and Overview of the Main Results}}
\label{summaryresults}

In many applications, prior knowledge of the data's structure is available. Incorporating this information appropriately into a reconstruction program has been shown to significantly improve the success of matrix recovery (\cite{SMC1,WMC3,WMC4,WMC12,WMC6,WMC8,WMC1,optinc,WMC7,WMC9,WMC11,ward,ward2}). Inspired by these approaches, this paper proposes and analyzes a matrix reconstruction framework that exploits prior knowledge of subspaces that align well with the matrix of interest. This section introduces the approach and summarizes the results and novelties of this paper.

The contribution of this work is in the generality of the proposed framework and analysis. The main result provides the ability to derive near-optimal sampling complexities and informative error bounds that express a trade-off when incorporating distinct subspaces enforcing known matrix structure. The result applies to a variety of \emph{weighted nuclear norm minimization} programs, providing novel insight to previous approaches in the literature and proposing new programs.

To introduce the approach and summarize the results, let $T\subset\mathbbm{C}^{n_1\times n_2}$ be a linear subspace of matrices with orthogonal complement $T^{\perp}$ and respective orthogonal projections $\mathcal{P}_{T}$ and $\mathcal{P}_{T^{\perp}}$. With weight parameter $0\leq\omega \leq 1$, the family of weighted nuclear norm minimization programs proposed here approximate the data matrix via the following modified version of (\ref{NNmin})
\begin{equation}
\label{wNNmins}
D^{\omega} \coloneqq \argmin\limits_{X\in\mathbbm{C}^{n_1\times n_2}}\big\|\omega\mathcal{P}_{T}(X)+\mathcal{P}_{T^{\perp}}(X)\big\|_* \ \ \mbox{subject to} \ \ \|P_{\Omega}(D) + d - P_{\Omega}(X)\|_2\leq\eta.
\end{equation}
When $\omega < 1$, program (\ref{wNNmins}) favors matrices that align with the estimate subspace $T$ while the case $\omega = 1$ reduces the program to unbiased nuclear norm minimization (\ref{NNmin}). The weight parameter toggles how severely one wishes to penalize matrices that do not agree with the prior information, thereby capturing the user's confidence in $T$. The main novelty of program \eqref{wNNmins} in contrast to previous approaches is in the ability to incorporate subspaces with general structure and the flexibility of weight selection. 

The theoretical contributions of this paper can be summarized as follows:

\begin{itemize}
    \item Theorem \ref{thm2} analyzes program (\ref{wNNmins}) in a general sense, applying to any subspace $T$ with elements of maximal rank $r$. The main result states that one can accurately reconstruct a matrix nearly lying in $T$ from  $m\geq r M_1 \polylog(n_1)$ observed entries, where $1\leq M_1\leq \frac{n_1n_2}{r}$ captures crucial dimensional and incoherence-based properties of $T$. The result is robust, with error bound proportional to the measurement noise level $\eta$, the error of the best rank $r$ approximation $\sum_{k=r+1}^{n_2}\sigma_k$, and a term that quantifies the accuracy of $T$.
    \item Specific choices of $T$ are presented in Section \ref{simpleWMCsec}, demonstrating the applicability of Theorem \ref{thm2}. The results derived therein showcase a variety of sampling complexities including $m\sim r\polylog(n_1)$, $m\sim r^2\polylog(n_1)$, and $m\sim (n_1r-r^2)\polylog(n_1)$. Furthermore, the derived error bounds express an informative trade-off between sampling complexity and sensitivity to the accuracy of $T$. In other words, programs incorporating subspaces $T$ that require less samples will in general exhibit error terms that are more susceptible to inaccurate $T$ (quantified via the principal angles between subspaces, \cite{angles}). This behavior is validated numerically in Section \ref{numexp}.
    \item Some examples in Section \ref{simpleWMCsec} are related to approaches previously studied in the literature: \cite{WMC9,WMC6,WMC8,WMC1,optinc,WMC7,WMC11}. Most of the methodologies or results from these citations only apply in a high fidelity scenario, when $T$ is error-free or aligns sufficiently well with the data matrix. In contrast, the main contribution here is the robustness of program (\ref{wNNmins}) and the theoretical results. The error bounds derived here provide novel insight to previous approaches, allowing inexact prior information while roughly matching the sampling complexity of related results in the literature. See Section \ref{discussion} for further discussion.
\end{itemize}

\subsection{Organization and Notation}
The remainder of the paper is organized as follows: Section \ref{WMCsec} discusses a foundational weighted matrix completion approach from the literature in order to elaborate on the inspiration for this work, its main result, and the improvements provided relative to the literature. Section \ref{simpleWMCsec} applies the main result to example subspaces $T$, deriving a variety of sampling complexities and error bounds. Section \ref{discussion} discusses related work in the literature and the introduced incoherence parameter $\mu_1$ in order to fairly compare this work with other results. Section \ref{numexp} conducts numerical experiments, comparing example programs to the original weighted program discussed in Section \ref{WMCsec}. The paper concludes with a discussion of future work in Section \ref{conclusion} followed by the proofs in the \href{mainproof}{Appendix}.

\underline{Notation:} for any integer $n\in\mathbbm{N}$, $[n]$ denotes the set $\{\ell\in\mathbbm{N}: 1 \leq \ell\leq n\}$ and $I_n$ is the $n\times n$ identity matrix. For $k,\ell\in\mathbbm{N}$, $b_{k}$ indicates the $k$-th entry of the vector $b$, $X_{k\ell}$ denotes the $(k,\ell)$ entry of the matrix $X$ and $X_{k*} \ (X_{*\ell})$ denotes its $k$-th row (resp. $\ell$-th column). For vectors, $\|b\|_2$ is the Euclidean norm. For matrices, $\sigma_k(X)$ denotes the $k$-th largest singular value of $X$, $\|X\| \coloneqq \sigma_1(X)$ is the operator norm, $\|X\|_F \coloneqq \langle X,X\rangle^{1/2}$ is the Frobenius norm, $\|X\|_* \coloneqq \sum_k\sigma_k(X)$ is the nuclear norm, and $\|X\|_{\infty}$ is the largest entry of $X$ in absolute value. $S$ and $S_{op}$ are the closed unit balls in $\mathbbm{C}^{n_1\times n_2}$ with respect to the Frobenius and operator norms respectively. The adjoint of a linear operator $\mathcal{A}$ is denoted by $\mathcal{A}^{*}$, while $X^{\top}$ will be used to denote the conjugate transpose of a matrix $X$. As previously mentioned, for matrices $X\in\mathbbm{C}^{n_1\times n_2}$ assume $n_1\geq n_2$ without loss of generality.

\section{\textcolor{black}{Weighted Matrix Completion}}
\label{WMCsec}

To the author's best knowledge, the first version of a weighted nuclear norm minimization program was proposed by \cite{SMC1}.  To elaborate on this original approach, given $r\leq n_2$, consider the SVD and decompose the data matrix of interest as
\[
D = U\Sigma V^{\top} = U^r\Sigma^r V^{r \top} + U^+\Sigma^+ V^{+ \top}
\]
where $U^r\Sigma^r V^{r\top}$ pertains to the largest $r$ singular values of $D$ along with the corresponding singular vectors. Assume that $\tilde{U}\in\mathbbm{C}^{n_1\times r}, \tilde{V}\in\mathbbm{C}^{n_2\times r}$ with orthonormal columns are available containing information of the range of $U^r\in\mathbbm{C}^{n_1\times r}, V^r\in\mathbbm{C}^{n_2\times r}$ and define
\[
Q_{\omega_1} \coloneqq \omega_1\tilde{U}\tilde{U}^{\top} + I_{n_1} - \tilde{U}\tilde{U}^{\top}, \ \ \ \ \ \ W_{\omega_2} \coloneqq \omega_2\tilde{V}\tilde{V}^{\top} + I_{n_2} - \tilde{V}\tilde{V}^{\top},
\]
where $\omega_1,\omega_2\in [0,1]$ are chosen weights. Notice that $Q_{1}, W_1$ are identity matrices and otherwise, when $\omega_1,\omega_2< 1$, these linear operators skew toward the orthogonal complement of range$(\tilde{U})$ and range$(\tilde{V})$ respectively. The original weighted nuclear norm minimization program approximates the data matrix via
\begin{equation}
\label{wNNmin2}
D^{\omega_1,\omega_2} \coloneqq \argmin\limits_{X\in\mathbbm{C}^{n_1\times n_2}}\|Q_{\omega_1}XW_{\omega_2}\|_* \ \ \mbox{subject to} \ \ \|P_{\Omega}(D) + d - P_{\Omega}(X)\|_2\leq\eta.
\end{equation}
Analogous to this paper's approach, with $\omega_1,\omega_2 < 1$ program (\ref{wNNmin2}) favors matrices that match a certain structure while $\omega_1=\omega_2 = 1$ is unbiased nuclear norm minimization (\ref{NNmin}). Notice that (\ref{wNNmin2}) is nearly of the form (\ref{wNNmins}), but the choice of two weights in the original formulation will impede the results here from being directly applicable. However, a program of the form (\ref{wNNmins}) that is closely related to (\ref{wNNmin2}) will be considered in Section \ref{simpleWMCsec}.

Spurring from the original program, similar methodologies have been proposed in the literature (\cite{WMC1,WMC3,WMC4,WMC6}) but distinct approaches have also been considered (\cite{WMC7,WMC8,WMC9,WMC11,WMC12,optinc}). To attempt producing a result that provides some level of insight for many of these variations, a more general notion of incoherence that applies to an entire subspace is required. 
\begin{definition}\label{def:subspace}
For a subspace $T\subset\mathbbm{C}^{n_1\times n_2}$ and $\rho\leq n_2$, the subspace $\rho$-incoherence and joint incoherence parameters of $T$ are defined respectively as
\begin{equation}
\label{Mu1}
M_0 := \max_{X\in T\cap S_{op}}\frac{n_2}{\rho}\|X\|_{\infty,2}^2
\end{equation}
and
\begin{equation}
\label{Mu2}
M_1 := \max_{X\in T\cap S}\frac{n_1n_2}{\rho}\max_{k,\ell}|X_{k\ell}|^2,
\end{equation}
where $\|X\|_{\infty,2}$ is the maximum of the row and column norms of $X$, see (\ref{infty2}).
\end{definition}
The ensemble will be referred to as the $\rho$-\emph{subspace incoherence} parameters or condition. These parameters will provide crucial dimensional information of a given subspace $T$. The role of these parameters will be discussed after the statement of the main result and further considerations are provided in Section \ref{simpleWMCsec}. Henceforth, let $\rho$ be defined as
\begin{equation}
\label{dimT}
\rho = \max_{X\in T} \ \mbox{rank}(X).
\end{equation}
The main result of the paper can now be presented:

\begin{theorem}
\label{thm2}
Let $T\subset\mathbbm{C}^{n_1\times n_2}$ be a subspace with $\rho$-subspace incoherence parameters $M_0, M_1$. Suppose $\Omega\subseteq [n_1]\times [n_2]$ is generated by selecting a subset of size $m\leq n_1n_2$ uniformly at random from all subsets of size $m$. Let $D\in\mathbbm{C}^{n_1\times n_2}$ and define $D^{\omega}$ as in (\ref{wNNmins}) with $\|d\|_2\leq \eta$. 
There exist universal constants $c_0,c_1,c_2,c_3,c_4>0$ such that if
\begin{equation}
\label{ocon1}
m\geq c_0M_1\rho\log^{3}(n_1) \ \ \ \ \mbox{and} \ \ \ \ 0<\omega \leq \frac{\sqrt{M_1}\log(n_1+n_2)}{\sqrt{n_1n_2}}
\end{equation}
or
\begin{equation}
\label{ocon2}
m\geq c_0\max\Big\{M_0,\sqrt{M_1}\Big\}\omega n_1\rho\log^{2}(n_1) \ \ \ \ \mbox{and} \ \ \ \ \frac{1}{2\sqrt{2\rho\log(n_1)}}\leq \omega \leq 1,
\end{equation}
then
\begin{align}
\label{errbd2}
&\|D-D^{\omega}\|_F \leq 
c_1f(\omega)\left(\sqrt{\frac{n_1n_2\log(n_1)}{m}} + \frac{1}{\omega}\right)\sum_{k=\rho+1}^{n_2}\sigma_k(D)  \nonumber \\
&+ c_2f(\omega)\left(\frac{1}{\omega} + \sqrt{\frac{n_1n_2\log(n_1)}{m}}\left(\omega\sqrt{\frac{n_1n_2\rho\log(n_1)}{m}} + 1\right)\right)\eta \nonumber \\
&+ c_3f(\omega)\left(\frac{1}{\omega} + \sqrt{\frac{n_1n_2\log(n_1)}{m}} \right) \big\|\mathcal{P}_{T^{\perp}}(U^{\rho}\Sigma^{\rho} V^{\rho \top})\big\|_*
\end{align}
holds with probability greater than $1-\frac{c_4}{n_1}$, where
\[
f(\omega) := \min\Bigg\{\omega\sqrt{\frac{n_1n_2\log(n_1)}{m}} , 1\Bigg\}.
\]
\end{theorem}

As in the beginning of the section, $U^{\rho}\Sigma^{\rho} V^{\rho \top}$ in \eqref{errbd2} denotes the best rank $\rho$ approximation of $D$ via the SVD. The proof is postponed until Section \href{2}{A.3}. The theorem implies that faithful prior subspace knowledge (with elements of maximal rank $\rho$) can be used to robustly approximate a nearly rank $\rho$ matrix in the subspace from as few as $M_1\rho\log^3(n_1)$ observed entries. The result applies to two intervals for $\omega$, with small and large weights ($\omega\approx 0$ and $\omega\approx 1$ respectively). These ranges offer distinct sampling complexities and error bounds. The main focus of this paper will be in the regime of small weights, which produce a substantial reduction in sampling complexity. However, the interval of larger weights \eqref{ocon2} also generates interesting results (see end of Section \ref{simpleWMCsec}).

The result is robust to inexact subspaces, expressed in the term
\[
\big\|\mathcal{P}_{T^{\perp}}(U^{\rho}\Sigma^{\rho} V^{\rho \top})\big\|_*
\]
appearing in (\ref{errbd2}). This term quantifies the inaccuracy of the estimate subspace relative to the best rank $\rho$ approximation of the data, where $\|\mathcal{P}_{T^{\perp}}(U^{\rho}\Sigma^{\rho} V^{\rho \top})\|_*\approx 0$ implies that $T$ was chosen appropriately. Arguably, this term is quite raw and the result can be more informative if a more general metric is used. The main convenience in presenting Theorem \ref{thm2} with error bound (\ref{errbd2}) is that no restrictions are imposed on $T$. 

To provide a more informative error bound (at the cost of restricting $T$), the accuracy of the estimate subspace can instead be quantified via the \emph{principal angle between subspaces} (PABS) \cite{Aangles}
\begin{equation}
\label{sinT}
    \sin(\theta_T) \coloneqq \big\|\mathcal{P}_{T^{\perp}}\mathcal{P}_{T_{\rho}}\big\|_{F\rightarrow F}.
\end{equation}
In (\ref{sinT}), $\|\circ\|_{F\rightarrow F}$ is the spectral norm for operators acting on matrices and $\mathcal{P}_{T_{\rho}}$ is the orthogonal projection onto a subspace of matrices that $U^{\rho}\Sigma^{\rho} V^{\rho \top}$ belongs to
\[
T_{\rho} \coloneqq \mbox{span}\Big\{U_{*k}^{\rho}V_{*k}^{\rho\top}\Big\}_{k=1}^{\rho}.
\]
Notice that the PABS lies in $[0,1]$, where $\sin(\theta_T)\approx 0$ implies accurate subspace estimate and $\sin(\theta_T)\approx 1$ states that $U^{\rho}\Sigma^{\rho} V^{\rho \top}$ largely lies in $T^{\perp}$. Since $\sin(\theta_T)\neq 1$ only when dim$(T)\geq$ dim$(T_{\rho})$ (see \cite{Aangles}), this restriction on $T$ must now be enforced in order to obtain the following modified version of Theorem \ref{thm2} involving the PABS:
\begin{corollary}
\label{corthm2}
Under the same conditions of Theorem \ref{thm2}, if \upshape{dim}$(T)\geq \rho$ \emph{then}
\begin{align*}
&\|D-D^{\omega}\|_F \leq 
c_1f(\omega)\left(\sqrt{\frac{n_1n_2\log(n_1)}{m}} + \frac{1}{\omega}\right)\sum_{k=\rho+1}^{n_2}\sigma_k(D) \\
&+ c_2f(\omega)\left(\frac{1}{\omega} + \sqrt{\frac{n_1n_2\log(n_1)}{m}}\left(\omega\sqrt{\frac{n_1n_2\rho\log(n_1)}{m}} + 1\right)\right)\eta \\
&+ c_3f(\omega)\sin(\theta_T)\left(\frac{1}{\omega} + \sqrt{\frac{n_1n_2\log(n_1)}{m}} \right) \left(n_2\sum_{k=1}^{\rho}\sigma_k^2(D)\right)^{1/2}
\end{align*}
\emph{holds with probability greater than} $1-\frac{c_4}{n_1}$.
\end{corollary}
Corollary \ref{corthm2} provides a more standard error bound, since the PABS are well understood metrics of signal alignment in the literature. However, aside from imposing dim$(T)\geq \rho$, notice that the last term in the error bound above now includes the avoidable factor $\sqrt{n_2}$. For this reason, the example results of Section \ref{simpleWMCsec} will apply Theorem \ref{thm2} rather than Corollary \ref{corthm2} since the latter seems to produce pessimistic error bounds in general.

\textcolor{black}{The results gauge the sampling complexity of the $T$-biased program \eqref{wNNmins} via the terms $M_1\rho$ and $\max\{M_0,\sqrt{M_1}\}n_1\rho$, depending on the weight choice (\ref{ocon1}) or (\ref{ocon2}) respectively. This complexity scales with dim$(T)$, but accounts for the suitability of matrices in $T$ to be probed in the matrix completion sense (i.e., incoherence). To see this, let $T$ be the span of $r$ linearly independent matrices, then $\rho \geq r$ and the bound $M_1 \geq 1$ requires that in (\ref{ocon1})
\[
m \geq M_1\rho \geq r = \mbox{dim}(T).
\]
Similarly, in \eqref{ocon2}, using $M_0\geq 1$ and the lower bound for $\omega$ gives
\[
m \geq n_1\max\{M_0,\sqrt{M_1}\}\rho\omega\log^2(n_1) \geq n_1\sqrt{\rho}\log^{3/2}(n_1) \geq \mbox{dim}(T).
\]
Therefore, the number of samples always adheres to the dimension of $T$, but may surpass this measure to ensure that incoherence conditions hold for the subspace of interest.}

To further discuss this observation and demonstrate the crucial role of $M_1$, it is instructive to consider specific choices of $T$. This is the purpose of Section \ref{simpleWMCsec}, where sampling complexities will be supplied for various estimate subspaces. The proofs of these results, which are corollaries of Theorem \ref{thm2}, consist of lower and upper bounding $M_1$. The lower bounds (presented in Appendix \href{corollaries}{B}) enforce sampling complexity upholding the dimensionality of $T$. These examples illustrate the importance of $M_1$ and the near optimality of the main result.

\subsection{\textcolor{black}{Matrix Completion with Prior Knowledge: Example Subspaces}}
\label{simpleWMCsec}

This section provides example programs that enforce specific subspace structure on the output estimate data matrix. Theorem \ref{thm2} will be applied in the small weight regime \eqref{ocon1} to derive informative reconstruction error bounds and nearly optimal sampling complexities. In particular, the results will showcase a trade-off when incorporating distinct choices of $T$. The examples will illustrate that subspaces allowing exact matrix completion from relatively few samples will in general exhibit increased sensitivity to inaccurate prior information.

To elaborate in a concrete setting, decompose as before $D = U\Sigma V^{\top} = U^r\Sigma^r V^{r \top} + U^+\Sigma^+ V^{+ \top}$, where $U^r\Sigma^r V^{r\top}$ is the best rank $r$ approximation of $D$. Assume that $\tilde{U}\in\mathbbm{C}^{n_1\times r}, \tilde{V}\in\mathbbm{C}^{n_2\times r}$ with orthonormal columns are available containing information of $U^r, V^r$. In this section, $\tilde{U}$ and $\tilde{V}$ will specify $T$. The accuracy of the prior information will be mainly quantified via the principal angle between subspaces (PABS) with respect to the range of the left and right singular vectors. From the SVD, notice that the columns of $U^+$ ($V^+$ respectively) form an orthonormal basis for range$(U^r)^{\perp}$ (range$(V^r)^{\perp}$ respectively). The PABS here are defined via the canonical correlation coefficients (\cite{angles})
\begin{equation}
\label{PABS}
\sin(\theta_u)\coloneqq \big\|\tilde{U}^{\top}U^+\big\| \ \ \mbox{and} \ \ \sin(\theta_v)\coloneqq \big\|\tilde{V}^{\top}V^+\big\|,
\end{equation}
\textcolor{black}{where $\|X\| = \sigma_1(X)$ denotes the operator norm of a matrix $X$}. The PABS lie in $[0,1]$ and provide a measure of the degree of alignment between the data's main rank $r$ component and estimates specified in some sense by $\tilde{U}$ and $\tilde{V}$. The case $\sin(\theta_u), \sin(\theta_v) \approx 0$ holds with accurate knowledge, while $\sin(\theta_u), \sin(\theta_v) \approx 1$ implies inappropriate prior information was supplied. 

\textcolor{black}{With this in mind, simplified results for four programs of the form (\ref{wNNmins}) with specific choices of $T$ are provided. Henceforth, program (\ref{wNNmins}) incorporating a subspace $T$ will be refered to as program (\ref{wNNmins}, $T$).} The following examples are presented in order of most to least sensitive error bounds with respect to the inaccuracy of $T$, respectively exhibiting a trend of increasing sampling complexity. 
\begin{theorem}
\label{cor1}
Let $D\in\mathbbm{C}^{n_1\times n_2}$ have rank $r$, and consider estimates $\tilde{U}\in\mathbbm{C}^{n_1\times r}$ and $\tilde{V}\in\mathbbm{C}^{n_2\times r}$ where $\tilde{U}\tilde{V}^{\top}$ has $r$-incoherence parameter $\mu_1$. Suppose $\Omega\subseteq [n_1]\times [n_2]$ is generated by selecting a subset of size $m\leq n_1n_2$ uniformly at random from all subsets of size $m$. Define $D^{\omega}$ as in (\ref{wNNmins}) with $\|d\|_2=\eta = 0$ and estimate subspace
\begin{equation}
T_1 \coloneqq \mbox{\upshape{span}}\Big\{\tilde{U}_{*k}\tilde{V}_{*k}^{\top}\Big\}_{k\in\{1,\cdots,r\}}. \nonumber
\end{equation}

There exists universal constants $c_0,c_3>0$ such that if
\[
m\geq c_0\mu_1r\log^{3}(n_1) \ \ \ \ \mbox{and} \ \ \ \ 0<\omega \leq \frac{\sqrt{\mu_1}\log(n_1+n_2)}{\sqrt{n_1n_2}}
\]
then with high probability
\begin{equation}
\label{errbds0}
\frac{\|D-D^{\omega}\|_F}{\|D\|_F} \leq c_3\sqrt{\frac{n_1n_2r\log(n_1)}{m}}\left(\sin(\theta_u) + \sin(\theta_v) + \left(\sum_{k\neq\ell}\big\| U^{r\top}\tilde{U}_{*k}\tilde{V}_{*\ell}^{\top}V^{r}\big\|_F^2\right)^{1/2}\right).
\end{equation}
\end{theorem}

To the author's best knowledge, \textcolor{black}{program (\ref{wNNmins}, $T_1$) provides a novel method by which to exploit prior information of the data's rank 1 components}. The result states that, with an appropriate weight and accurate projection onto the subspace spanned by data's rank 1 components, program (\ref{wNNmins}, $T_1$) can faithfully estimate a rank $r$ matrix with $\mathcal{O}(r\log^3(n_1))$ random samples. In a high fidelity scenario, there are roughly $r$ degrees of freedom to approximate $U^r\Sigma^rV^{r\top}$. Therefore the derived sampling complexity is within a quadratic logarithmic factor of the optimal rate $\mathcal{O}(r\log(n_1))$, where the logarithmic factor in the optimal rate is unavoidable in matrix completion under random sampling models (see for example Theorem 1.7 provided by \cite{MC6}). On the other hand, in the case of unreliable information, the right hand side of (\ref{errbds0}) showcases high sensitivity to inaccurate $T_1$. Several elaborated remarks are in order:

\begin{itemize}
\item Incorporating accurate subspace $T_1$ allows for frugal completion of matrices that do not satisfy typical incoherence conditions. For example, with severe parameter $\mu_0\sim \sqrt{n_1}/r$, since $\mu_1\leq\mu_0^2r$ Theorem \ref{cor1} guarantees reconstruction from $\mathcal{O}(n_1\log^3(n_1))$ random samples. Without prior information, such data matrices require a significant amount of additional samples to be recovered according to Theorem \ref{thm1} and similar results in the literature.

\item Program (\ref{wNNmins}, $T_1$) allows reconstruction of general full rank matrices in an underdetermined scenario with prior information. With $n_2$ trustworthy rank 1 components (or an orthogonal projection onto the span), the result allows for an accurate estimate of a full rank matrix from $\mathcal{O}(n_2\log^3(n_1))$ observed entries.

\item In contrast to other approaches that will be considered, program (\ref{wNNmins}, $T_1$) is most sensitive to inaccurate prior information. This observation will become clear due to the final term in \eqref{errbds0} involving the sum over all $k\neq\ell$, which is not present in the remaining error bounds. This term requires that each matrix $\tilde{U}_{*k}\tilde{V}_{*\ell}^{\top}$ with $k\neq\ell$ be unaligned with elements in 
\[
\mbox{\upshape{span}}\Big\{U_{*k}^rV_{*\ell}^{r\top}\Big\}_{k,\ell\in\{1,\cdots,r\}}.
\]
This requisite is quite strict, since even a single inaccurate rank-1 component included as prior knowledge can cause a significant error according to \eqref{errbds0}. 

\end{itemize}

The next result involves a methodology related to many previously proposed in the literature (\cite{optinc,WMC7,WMC8,WMC9,WMC11,WMC12}). The result will render a less sensitive methodology at the cost of higher sampling complexity.
\begin{theorem}
\label{cor2}
Under the same setup as Theorem \ref{cor1}, let $\tilde{U}\tilde{U}^{\top}$ and $\tilde{V}\tilde{V}^{\top}$ have $r$-standard incoherence parameters $\mu_L$ and $\mu_R$ respectively. Define $D^{\omega}$ as in (\ref{wNNmins}) with $\eta = 0$ and subspace estimate
\begin{equation}
T_2 \coloneqq \mbox{\upshape{span}}\Big\{\tilde{U}_{*k}\tilde{V}_{*\ell}^{\top}\Big\}_{(k,\ell)\in\{1,\cdots,r\}\times\{1,\cdots,r\}}. \nonumber
\end{equation}

There exists universal constants $c_0,c_3>0$ such that if
\[
m\geq c_0\mu_{L}\mu_{R}r^2\log^{3}(n_1) \ \ \ \ \mbox{and} \ \ \ \ 0<\omega \leq \frac{\sqrt{\mu_{L}\mu_{R} r}\log(n_1+n_2)}{\sqrt{n_1n_2}}
\]
then with high probability
\begin{equation}
\label{errbds1}
\frac{\|D-D^{\omega}\|_F}{\|D\|_F} \leq c_3\left(\sin(\theta_u) + \sin(\theta_v)\right)\sqrt{\frac{n_1n_2r\log(n_1)}{m}}.
\end{equation}
\end{theorem}

This result considers a robust approach in contrast to program (\ref{wNNmins}, $T_1$), with a more lenient requisite that the columns of $U^r$ and $V^r$ lie in the range of $\tilde{U}$ and $\tilde{V}$ respectively. This comparison is evident in the produced error bound \eqref{errbds1}, which is similar to \eqref{errbds0} but no longer exhibits the sensitive term that sums over $k\neq\ell$. However, the trade-off is an increased sampling complexity of $\mathcal{O}(r^2\log^3(n_1))$. When $\omega\approx 0$, this methodology and result roughly agree with the work by \cite{WMC9,optinc,WMC7,WMC11}, but program (\ref{wNNmins}, $T_2$) is more flexible since it introduces a choice of weight and allows for erroneous row and column subspaces (for further discussion on related work see Section \ref{relatedwork}).

Theorems \ref{cor1} and \ref{cor2} help illustrate the crucial role of the $\rho$-subspace incoherence parameters from Definition \ref{def:subspace}. Consider program (\ref{wNNmins}, $T_2$), where $\rho = r$. To obtain Theorem \ref{cor2} from Theorem \ref{thm2}, it will be shown in Appendix \href{corollaries}{B} that $M_1(T_2) = \mu_L\mu_R r$. Since $T_2$ is an $r^2$ dimensional space, one should not expect exact matrix recovery via program (\ref{wNNmins}, $T_2$) with less than $r^2\log(n_1)$ randomly sampled entries. This sampling complexity is enforced by the lower bound $\mu_L\mu_R r\geq r$, which illustrates the near optimality of the main result. Similarly, for program (\ref{wNNmins}, $T_1$) it holds that $M_1(T_1) = \mu_1 \geq 1$ which upholds the dimensionality of $T_1$.

The next example involves only incorporating row or column span information.
\begin{theorem}
\label{cor4}
Under the same setup as Theorem \ref{cor1}, let $\tilde{U}\tilde{U}^{\top}$ have $r$-standard incoherence parameter $\mu_L$. Define $D^{\omega}$ as in (\ref{wNNmins}) with $\eta = 0$ and subspace estimate
\begin{equation}
T_3 \coloneqq \Big\{X\in\mathbbm{C}^{n_1\times n_2} \ | \ X = \tilde{U}\tilde{U}^{\top}X\Big\}. \nonumber
\end{equation}

There exists universal constants $c_0,c_3>0$ such that if
\[
m\geq c_0\mu_Ln_2r\log^{3}(n_1) \ \ \ \ \mbox{and} \ \ \ \ 0<\omega \leq \frac{\sqrt{\mu_L}\log(n_1+n_2)}{n_1}
\]
then with high probability
\begin{equation}\label{errbds2}
    \frac{\|D-D^{\omega}\|_F}{\|D\|_F} \leq c_3\sin(\theta_u)\sqrt{\frac{n_1n_2r\log(n_1)}{m}}.
\end{equation}
\end{theorem}
In contrast to unbiased nuclear norm minimization (\ref{NNmin}), the result demonstrates that one sided information can reduce the sampling complexity for rectangular matrices with $n_2 \ll n_1$. Analogously, for wide matrices one can incorporate information of the range instead. Comparing to the previous examples (incorporating $T_1$ and $T_2$), notice that error bound \eqref{errbds2} is less sensitive to inaccurate subspaces as it only involves a single PABS term.

The final example attempts to provide some intuition for the original weighted nuclear norm approach. Although Theorem \ref{thm2} is not directly applicable to program (\ref{wNNmin2}) and variations, some intuition may be provided by considering the program of the form (\ref{wNNmins}) with estimate subspace
\begin{equation}
\label{T4}
T_4 := \Big\{X\in\mathbbm{C}^{n_1\times n_2} \ | \ X = \tilde{U}\tilde{U}^{\top}X\tilde{V}\tilde{V}^{\top} + \tilde{U}\tilde{U}^{\top}X\tilde{V}^{\perp}\tilde{V}^{\perp \top} + \tilde{U}^{\perp}\tilde{U}^{\perp \top}X\tilde{V}\tilde{V}^{\top} \Big\}
\end{equation}
where the orthonormal columns of $\tilde{U}^{\perp}$ (resp. $\tilde{V}^{\perp}$) span range$(\tilde{U})^{\perp}$ (resp. range$(\tilde{V})^{\perp}$). Arguably, among all programs considered here, this approach is most related to the original weighted program. To gain insight, a new notion of incoherence is needed.
\begin{definition} 
\label{coherence2}
Given $D\in\mathbbm{C}^{n_1\times n_2}$ and $r\leq\min\{n_1,n_2\}$, consider the singular value decomposition (SVD) $D = U\Sigma V^{\top}$. The $r$-\emph{complementary incoherence parameter} of $D$ is defined as the largest $\mu_2$ such that
\begin{equation}
\label{munew}
\min_{1\leq k \leq n_1}\sqrt{\sum_{j=1}^{r}\big|U_{kj}\big|^2}\geq \sqrt{\frac{\mu_2r}{n_1}}, \ \ \ \ \ \ \min_{1\leq \ell \leq n_2}\sqrt{\sum_{j=1}^{r}\big|V_{\ell j}\big|^2}\geq \sqrt{\frac{\mu_2r}{n_2}}.
\end{equation}
\end{definition}
This incoherence condition is strongly related with the standard incoherence parameter. Notice that $1 = \|\tilde{U}_{p*}\|_2^2 + \|\tilde{U}_{p*}^{\perp}\|_2^2$ for any $p\in [n_1]$ and the same observation holds for $\tilde{V}$. Therefore, $0\leq\mu_2\leq 1\leq \mu_0$ and for this reason the new condition is referred to as complementary. This new parameter also quantifies how spiky a data matrix is, where $\mu_2 \approx 0$ holds for matrices with entire rows or columns containing very little information and $\mu_2 \approx 1$ implies an even distribution of non-zero entries throughout the matrix. The result for the weighted program incorporating (\ref{T4}) now follows.
\begin{theorem}
\label{cororiginal}
Under the same setup as Theorem \ref{cor1}, let $\tilde{U}\tilde{V}^{\top}$ have $r$-standard and complementary incoherence parameters $\mu_0$ and $\mu_2$ (resp.). Define $D^{\omega}$ as in (\ref{wNNmins}) with $\eta = 0$ and estimate subspace $T_4$ from (\ref{T4}).

There exists universal constants $c_0,c_3>0$ such that if
\[
m\geq c_0\mu_0\max\{\mu_0r,n_1-\mu_2r\}r\log^{3}(n_1) \ \ \ \ \mbox{and} \ \ \ \ 0<\omega \leq \frac{\sqrt{n_1-\mu_2r}\log(n_1+n_2)}{\sqrt{n_1n_2}}
\]
then with high probability
\begin{equation}\label{errbds3}
    \frac{\|D-D^{\omega}\|_F}{\|D\|_F} \leq c_3\sin(\theta_u)\sin(\theta_v)\sqrt{\frac{n_1n_2r\log(n_1)}{m}}.
\end{equation}
\end{theorem}
Under lenient conditions, it holds that $\max\{\mu_0r,n_1-\mu_2r\} = n_1-\mu_2r$ and therefore a $n_1\times n_2$ rank $r$ matrix can be recovered from $\sim \mu_0(n_1r-\mu_2r^2)\log^3(n_1)$ observed entries with prior information. Notice that $T_4$ is an $n_1r + n_2r-r^2$ dimensional subspace, which is respected by the complexity of the result. In contrast to unbiased nuclear norm minimization, this approach provides a slight reduction in sampling complexity by requiring $\mu_0\mu_2r^2\log^3(n_1)$ less samples. However, the current result exhibits a larger logarithmic dependency while removing the incoherence condition related to $\mu_1$, so a fair comparison is difficult to make. 

Relative to the other approaches considered in this section, program (\ref{wNNmins}, $T_4$) is most robust to inaccurate prior information since it imposes less of a constraint on the search space. This is exhibited in \eqref{errbds3} via the term $\sin(\theta_u)\sin(\theta_v)$, which is the smallest error bound in the examples thusfar. However, this reduction in sensitivity to noise produces one of the largest sampling requisites of this section. The trend expressed by these error bounds will be explored numerically in Section \ref{numexp}, where decreasing susceptibility to inaccurate prior information in general requires an increasing number of observed entries for exact matrix completion.

The results in this section all apply Theorem \ref{thm2} in the regime of small weights (\ref{ocon1}). The context of larger weights (\ref{ocon2}) can also supply interesting results. For example, program (\ref{wNNmins}, $T_1$) with $\omega = (r\log(n_1))^{-1/2}$ can be shown to require $\max\{\mu_0,\sqrt{\mu_1}\}n_1\sqrt{r}\log^{3/2}(n_1)$ samples for accurate matrix estimation. Similarly, program (\ref{wNNmins}, $T_2$) would need $\mu_0n_1r\log^{3/2}(n_1)$ observed entries, mildly reducing the sampling complexity in contrast to (\ref{NNmin}) by a logarithmic factor. Intuitively, such larger weight choices would be most appropriate for less reliable prior information while still reducing the number of samples.

\section{Discussion}
\label{discussion}

This section elaborates on the main result and the considered weighted programs. Comparison to previous work is conducted in Section \ref{relatedwork} and a discussion of the novel incoherence parameters is provided in Section \ref{incoherence}.

\subsection{Related Work}
\label{relatedwork}

Many authors have considered how to efficiently include prior information into matrix reconstruction problems. The papers by \cite{SMC1,WMC3,WMC4} focus on applications to seismology and numerical aspects of the problem, adopting an approach that spurred from program (\ref{wNNmin2}) first proposed by \cite{SMC1}. A distinct approach is considered by \cite{ward,ward2,MC_2}, where the prior information is used to bias the sampling scheme according to the array's leverage scores. Therein, the authors show that $\mathcal{O}(n_1r\log^2(n_1))$ revealed entries provide exact (noiseless) completion of a rank $r$ matrix with more lenient dependency on the standard incoherence parameter $\mu_0$. Most relevant to the context of this paper, are the results by \cite{WMC9,WMC6,WMC8,WMC1,optinc,WMC7,WMC11} which provide theoretical analysis related to program (\ref{wNNmin2}) or program (\ref{wNNmins}, $T_2$) in Theorem \ref{cor2}.

The work by \cite{WMC1} applies directly to program (\ref{wNNmin2}) in the matrix completion case and general matrix sensing scenario. In the matrix completion setting the authors obtain sampling complexity $|\Omega|\sim \mu_0n_1r\log(n_1)$ under inexact prior information, thereby reducing the number of samples by a logarithmic factor in contrast to unbiased nuclear norm minimization. The sampling complexity and error bounds produced therein depend on the PABS \eqref{PABS}, imposing a fidelity requisite on $T$ for the results to be applicable. The results in this paper are not directly comparable. Arguably, the result most related in this work is provided in Theorem \ref{cororiginal}, where program (\ref{wNNmins}, $T_4$) requires $|\Omega|\sim \mu_0(n_1-\mu_2r)r\log^{3}(n_1)$ samples. In contrast to the work by \cite{WMC1}, the result here does not require PABS-based prior knowledge or requisites for applicability and the resulting error bounds are more lenient in terms of the dependence on $\sin(\theta_u)$ and $\sin(\theta_v)$. The authors \cite{WMC6} consider a more flexible program than the original approach, where four weight choices are allowed. Their results only apply to the matrix sensing scenario, where the authors demonstrate an $\mathcal{O}(nr)$ sampling complexity.

The remaining citations (\cite{WMC9,optinc,WMC7,WMC11,WMC12}) consider approaches arguably similar in nature to program (\ref{wNNmins}, $T_2$) with $\omega = 0$ and exact prior information, while \cite{WMC8} allows for noisy side information. Among these, the smallest sampling complexity is $\mathcal{O}(\mu_0r^2\log(n_1)\log(r))$ provable in the setting of exact prior knowledge (in the analogous scenario where $\rho = r$ for simplicity). In this context, Theorem \ref{cor2} presented here with small weights allows for accurate estimation from $\mathcal{O}(\mu_0^2r^2\log^3(n_1))$ sampled entries. This sampling condition is slightly worse than what is derived by \cite{WMC9,optinc}. However, the approach considered here is a more flexible methodology with weight selection that allows for improved reconstruction output (see Section \ref{numexp} for numerical behaviour based on weight selection). Furthermore, Theorem \ref{cor2} applies to cases with inexact prior knowledge and derives error bounds that provide insight when inaccurate estimates $\tilde{U}$ and $\tilde{V}$ are incorporated in different ways.

\subsection{\textcolor{black}{Incoherence Parameter $\mu_1$}}
\label{incoherence}
This section discusses the parameter $\mu_1$ defined in (\ref{mu2}), comparing it to the standard and joint incoherence parameters from the literature. To the best of the author's knowledge, this definition of incoherence has not appeared in the matrix completion literature. Previous optimal results have sampling complexity requiring only linear dependence on the standard parameter $\mu_0$. The results here also depend on this parameter, but additionally introduce $\mu_1$ with sub-linear and linear dependence. 

Arguably, $\mu_1$ is reminiscent of the joint incoherence (or strong incoherence) condition introduced by \cite{MC2,MC3}. The joint incoherence parameter will be denoted here as $\tilde{\mu}_1$. Given $r\leq n_2$ and recalling the SVD of $D$, the joint incoherence parameter of $D$ is defined as the smallest $\tilde{\mu}_1>0$ such that
\[
\max_{(k,\ell)\in[n_1]\times[n_2]}\Bigg|\sum_{j=1}^{r}U_{kj}\bar{V}_{\ell j}\Bigg|\leq \sqrt{\frac{\tilde{\mu}_1r}{n_1n_2}}.
\]
In particular, $\mu_1$ in definition (\ref{mu2}) also depends jointly on the right and left singular vectors. Furthermore, $\mu_1\leq \mu_0^2r$ which is also a tight upper bound for the joint incoherence parameter (\cite{optinc}). However, it is important to note that \cite{MC2} and other authors derived sampling complexity $m\sim \tilde{\mu}_1n_1r\log^2(n_1)$. In contrast, the work here requires $m\sim \sqrt{\mu_1}n_1r\log^2(n_1)$ in Theorem \ref{thm1} and $m\sim \mu_1r\log^3(n_1)$ in Theorem \ref{cor1}. Though it is difficult to provide a fair comparison, this section will argue that the results here impose relatively lenient incoherent conditions in a ``joint'' sense.

The author \cite{optinc} discusses the exorbitant nature of the joint incoherence parameter since it intuitively requires the rows of $U^r$ and $V^r$ to be unaligned, a requisite with no reasonable explanation. In the current work, all results would also hold if $\mu_1$ were defined as the smallest number such that
\begin{equation}
\label{mu3}
\max_{1\leq k\leq n_1}\left(\sum_{j=1}^{r}|U_{kj}|^4\right)^{1/4}\max_{1\leq\ell\leq n_2}\left(\sum_{j=1}^{r}|V_{\ell j}|^4\right)^{1/4} \leq \sqrt{\frac{\mu_1r}{n_1n_2}}.
\end{equation}
This alternative definition alleviates the joint nature of the original definition. It is now arguable that this condition requires the $\ell_4$ norms of the rows of $U^r$ and $V^r$ to both be small, thereby imposing an additional ``non-spikiness'' condition analogous to the requisite of a small $\mu_0$ parameter with respect to the $\ell_2$ norms. However, in contrast the adopted definition (\ref{mu2}), the incoherence parameter given by (\ref{mu3}) is pessimistic (see the second example below) and for this reason the original definition is kept.

Notice that Theorem \ref{thm1} only has sub-linear dependence on the introduced parameter $\sqrt{\mu_1}$. This observation is crucial to properly compare some of the work here with previous incoherence optimal results. To elaborate, specific data matrices are produced to compute $\mu_0$, $\sqrt{\mu_1}$, and $\tilde{\mu}_1$.

\begin{itemize}
\item \textbf{Case $\mu_0 = \sqrt{\mu_1}$:} consider the random orthogonal model and the incoherent basis model by \cite{MC4}. Therein, the authors show that a rank $r$ matrix $M=U\Sigma V^{\top}$ generated from the random orthogonal model obeys $\max_{k,\ell}|U_{k\ell}|^2\leq 10\log(n_1)/n_1$ and $\max_{k,\ell}|V_{k\ell}|^2\leq 10\log(n_2)/n_2$ with high probability. From this, it is easy to see that $\mu_0,\sqrt{\mu_1} \sim \log(n_1)$. A similar conclusion holds trivially for the incoherent basis model and any singular vectors obeying the size property 1.12 in the same reference.

Joint incoherence: in this example, $\tilde{\mu}_1 \sim \log^2(n_1)$ (see \cite{MC4}).

\item \textbf{Case $\mu_0 > \sqrt{\mu_1}$:} let $r<n_2$ and $M=U V^{\top}$ where the columns of $U$ and $V$ consist of any $r$ columns of $I_{n_1}$ (the $n_1\times n_1$ identity matrix) and any $r$ columns of $\mathcal{F}:\mathbbm{C}^{n_2}\mapsto\mathbbm{C}^{n_2}$ (the 1D Fourier transform) respectively. Then the $r$-incoherence parameters satisfy $\mu_0 = n_1/r$ and $\sqrt{\mu_1} = \sqrt{n_1/r}$. Note that definition (\ref{mu3}) would instead obtain $\sqrt{\mu_1} = \sqrt{n_1/\sqrt{r}}$ in this example, which demonstrates the improvement gained by the chosen definition (\ref{mu2}).

Joint incoherence: here $\tilde{\mu}_1 \sim n_1/r$.

\item \textbf{Case $\mu_0 < \sqrt{\mu_1}$:} let $M=U V^{\top}$ where the columns of $U$ and $V$ consist of any $r$ columns of $I_{n_1}$ and any $r$ columns of $I_{n_2}$ respectively. This example gives the worst case $r$-incoherence parameters $\mu_0 = n_1/r$ and $\sqrt{\mu_1} = \sqrt{n_1n_2/r}$. In general, using (\ref{mu3}) shows that $\sqrt{\mu_1}\leq \mu_0\sqrt{r}$, which is sharp according to this example when $n_1=n_2$.

Joint incoherence: $\tilde{\mu}_1 \sim n_1n_2/r$.

\end{itemize}

From these examples, it is clear that there is no strict relationship between $\mu_0$ and $\sqrt{\mu_1}$. Moreover, typical data matrices of interest from the literature seem to largely lie in the regime where $\sqrt{\mu_1}\leq \mu_0$. In these cases, Theorem \ref{thm1} intuitively reduces to solely depend on $\mu_0$.

Moreover, the examples reveal that $\mu_1\sim\tilde{\mu}_1$ holds intuitively, though a proof of such a statement is not provided in this work. The lenient joint incoherence conditions of this paper are due to the sub-linear dependence $\sqrt{\mu_1}$ in Theorem \ref{thm1}, which is most comparable to the work of \cite{MC2}. Notice that $\sqrt{\mu_1}\leq\tilde{\mu}_1$ holds in the examples above. However, Theorem \ref{cor1} does require linear dependence on $\mu_1$ (but removes the linear dependence on $n_1$). Since $\mu_1$ may be as large as $r\mu_0^2$, this may lead to the requisite $m\sim r^2\mu_0^2\log^3(n_1)$ which still offers a severe reduction in the number of observed entries comparable to the result in Theorem \ref{cor2}.

\section{Numerical Experiments}
\label{numexp}

This section numerically explores programs (\ref{wNNmins}, $T_1$) and (\ref{wNNmins}, $T_2$), comparing them to the original weighted nuclear norm minimization program (\ref{wNNmin2}). The goal of this section is to numerically validate the error bounds in Section \ref{simpleWMCsec}. The experiments will reveal the practicality of the derived analysis, agreeing with the theoretical observation that weighted programs incorporating subspaces that require less samples will in general exhibit increasing sensitivity of inaccurate prior information.

The setup of \cite{WMC1} is adopted to generate a data matrix and subspace information. Let $D = U^r\Sigma^r V^{r\top}\in\mathbbm{R}^{n_1\times n_2}$, where $U^r\in\mathbbm{R}^{n_1\times r}$ and $V^r\in\mathbbm{R}^{n_2\times r}$ are constructed by orthogonalizing the columns of a standard random Gaussian matrix with $r$ columns and normalizing so that $\|D\|_F = 1$. To obtain prior knowledge, a perturbed matrix is generated $\tilde{D} = D + N$ where the entries of $N\in\mathbbm{R}^{n_1\times n_2}$ are i.i.d. Gaussian random  variables with variance $\sigma^2$ that will be toggled to select a desired PABS. Then $\tilde{U}\in\mathbbm{R}^{n_1\times r}$ and $\tilde{V}\in\mathbbm{R}^{n_2\times r}$ are the leading $r$ left and right singular vectors of $\tilde{D}$. The dimensions are set to $n_1 = n_2 = 500$ and $r = 50$. The set of observed matrix entries is selected uniformly at random from all subsets of the same cardinality $|\Omega| = \lambda (n_1n_2)$, where $\lambda\in [0,1]$ will be varied to specify a desired sampling percentage. In each experiment, $D, N$ and $\Omega$ are generated independently and programs (\ref{wNNmins}) and (\ref{wNNmin2}) are solved with $\omega = \omega_1\omega_2$ varying in (0,1] (setting $\omega_1=\omega_2$). The plots below present the average relative errors of 100 independent trials via trustworthy and relatively inaccurate subspace estimates.

Programs (\ref{wNNmins}) and (\ref{wNNmin2}) are solved using the LR-BPDN implementation introduced by \cite{SMC1}, which combines the Pareto curve methodology (\cite{pareto}) with a matrix factorization approach. With $\omega > 0$, (\ref{wNNinv2}) is solved in lieu of (\ref{wNNmins}), which is an equivalent formulation that trades off the objective function with a modified projection operator in the constraint (see \href{SecWMC}{A.3}). This allows LR-BPDN to be directly applicable to (\ref{wNNinv2}), with output $\tilde{D}^{\omega} = \omega\mathcal{P}_T(D^{\omega})+\mathcal{P}_{T^{\perp}}(D^{\omega})$ which gives the desired solution as $D^{\omega} = \omega^{-1}\mathcal{P}_T(\tilde{D}^{\omega})+\mathcal{P}_{T^{\perp}}(\tilde{D}^{\omega})$. An analogous trick is used to solve (\ref{wNNmin2}) when $\omega_1\omega_2 > 0$.

\underline{Varying weights:} noiseless numerical results with varying weights $\omega,\omega_1\omega_2\in (0,1]$ are shown in Figure \ref{wgt}. Two plots are provided, corresponding to reliable prior information with $|\Omega|/n_1n_2 = .01$ (left plot) and less accurate prior knowledge with $|\Omega|/n_1n_2 = .15$ (right plot). In these plots, the variance of $N$ is chosen to provide PABS $\sin(\theta_u),\sin(\theta_v)\approx .1$ and $\sin(\theta_u),\sin(\theta_v)\approx .2$ respectively. 

In the case of good subspace estimates (left plot) it is clear that program (\ref{wNNmins}, $T_1$) greatly outperforms the other approaches, obtaining a relative error $\approx .1$ with only $1\%$ of observed entries. However, this approach is demonstrated to be relatively sensitive when less accurate prior information is supplied. With relatively inaccurate prior information, the original weighted program (\ref{wNNmin2}) exhibits the best reconstruction error in the right plot of Figure \ref{wgt}. The numerical behavior illustrated in Figure \ref{wgt} agrees with the derived error bounds and discussion provided in Section \ref{simpleWMCsec}, where Theorem \ref{cororiginal} provides some insight for the robust behavior of program \eqref{wNNmin2} at the cost of higher sampling complexity.

\begin{figure}[!htb]
\centering
\includegraphics[width=0.45\hsize]{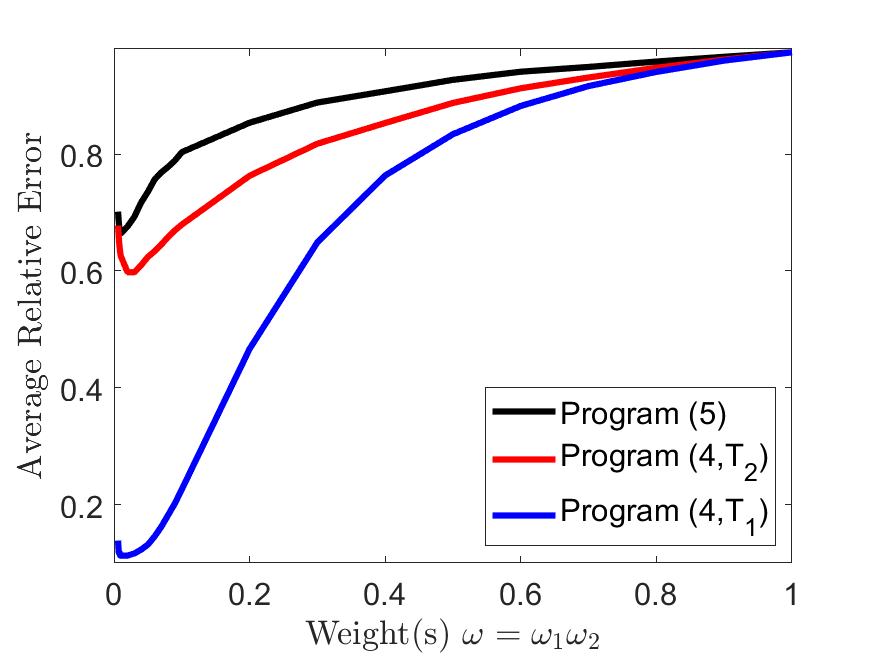} \ \ \ \ \ \ \
\includegraphics[width=0.45\hsize]{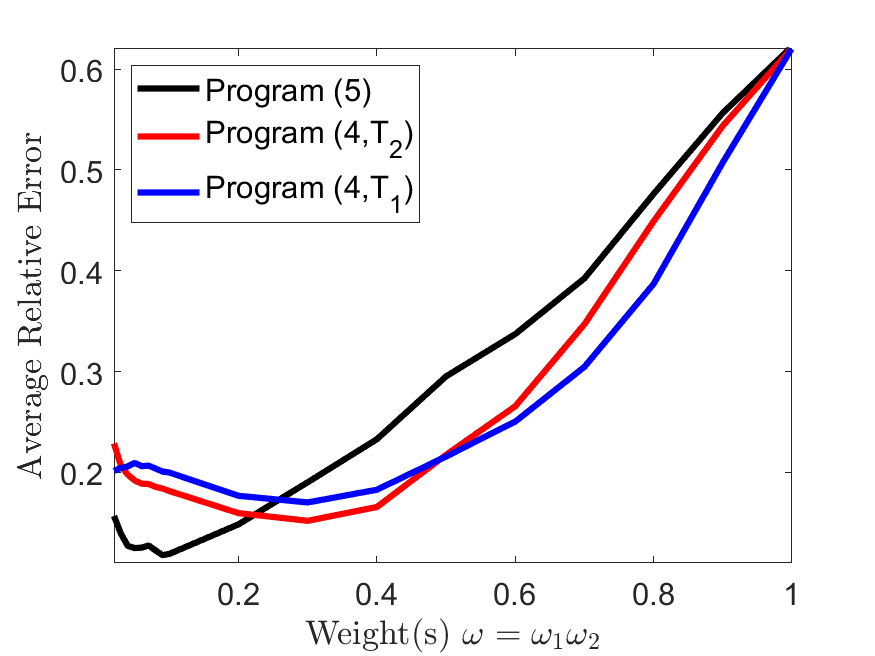}
\caption{Plots of weight choice versus average relative error of matrix reconstruction via noiseless weighted nuclear norm minimization programs. The left plot applies reliable prior information with $\sin(\theta_u),\sin(\theta_v)\approx .1$ and $\% 1$ percent of observed matrix entries. The right plot was obtained with $\sin(\theta_u),\sin(\theta_v)\approx .2$, where $\% 15$ of the entries are observed.}\label{wgt}
\end{figure}

\underline{Varying sampling percentage:} noiseless numerical results with varying percentages of observed noiseless entries are shown in Figure \ref{perc}. Applying the choice of weights from Figure \ref{wgt} that give the smallest reconstruction error for each program, two plots are shown with reliable subspaces (left plot) and less accurate prior knowledge (right plot). In the case of accurate estimate subspaces, program (\ref{wNNmins}, $T_1$) obtains the smallest relative errors in all shown sampling percentages. However, observe that programs (\ref{wNNmins}, $T_2$) and (\ref{wNNmin2}) obtain comparable relative errors from as little as $8\%$ of observed entries. As in Figure \ref{wgt}, in the case of less accurate subspaces, the original weighted program exhibits the most flexibility toward untrustworthy prior information. Therefore, Figure \ref{perc} also reflects the trade-off behavior of sampling complexity vs sensitivity to inaccurate prior information that agrees with the theoretical conclusions of Section \ref{simpleWMCsec}.

\begin{figure}[!htb]
\centering
\includegraphics[width=0.45\hsize]{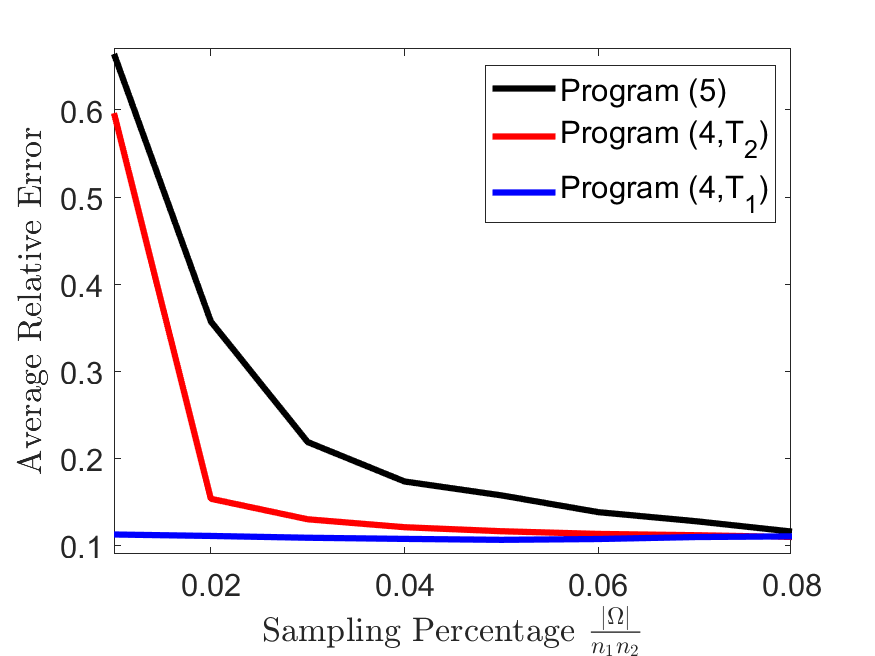} \ \ \ \ \ \ \
\includegraphics[width=0.45\hsize]{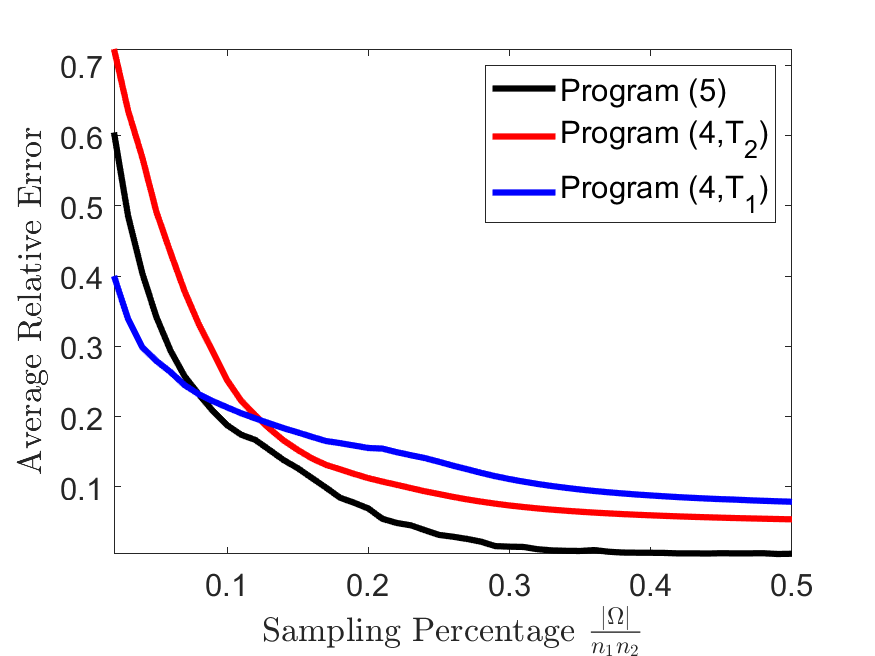}
\caption{Plots of sampling percentage versus average relative error of matrix reconstruction via noiseless weighted nuclear norm minimization programs. The left plot applies reliable prior information with $\sin(\theta_u),\sin(\theta_v)\approx .1$ and weight choices of $\omega_1\omega_2 = .01$ for program \ref{wNNmin2}, $\omega = .02$ for program (\ref{wNNmins}, $T_2$), and $\omega = .01$ for program (\ref{wNNmins}, $T_1$). The right plot was obtained with $\sin(\theta_u),\sin(\theta_v)\approx .2$ and weight choices of $\omega_1\omega_2 = .06$ for program \ref{wNNmin2}, $\omega = .3$ for program (\ref{wNNmins}, $T_2$), and $\omega = .2$ for program (\ref{wNNmins}, $T_1$).}\label{perc}
\end{figure}

\section{Conclusion}
\label{conclusion}

In this paper, a family of weighted matrix completion programs is proposed as a means to incorporate prior knowledge into the matrix completion problem. The main result establishes the nearly optimal sampling rate $\mathcal{O}(M_1 r\log^3(n))$ by which an $n\times n$ rank $r$ matrix can be accurately approximated when a subspace of rank $r$ matrices is enforced (where $M_1$ captures dimensional and incoherence-based properties of the subspace). The analysis allows for robust matrix approximation in the case of inexact subspaces, measurement noise, and full rank data matrices. The work provides novel intuition for previous methodologies in the literature, while introducing novel approaches. Finally, the results and numerical experiments showcase an insightful trade-off caused by incorporating distinct subspaces. In general, it is observed that subspaces requiring less samples for exact matrix completion are more susceptible to inaccurate prior information.

Several important limitations and potential improvements need to be elaborated and explored as future work. The subspace incoherence parameter $M_1$ may extract the parameter $\mu_1$, which is distinct in comparison to results that only require the standard incoherence condition. Modified analysis that only depends on the standard parameter $\mu_0$ would be most informative and align best with the existing literature. Furthermore, the main theoretical result does not seem to provide any insight for weight selection according to the user's confidence in the estimate subspace or any other variables. An error bound that specifies the optimal weight choice in different scenarios would be of great value to practitioners. As future work, it would be of interest to adapt the proof strategy and tools of this work to further comprehend how to set the program parameters of the proposed method and previously considered weighted approaches.

\textcolor{black}{Another avenue to extend this work would be to consider numerically efficient alternatives that exploit known matrix structure. For example, with strong prior information, a least squares-based weighted program might arguably produce comparable estimates while reducing the computational complexity. The tools developed in this work may be useful to analyze such approaches, in order to determine the benefits and disadvantages in contrast to the weighted nuclear norm minimization program considered here.}

\appendix
\section*{Appendix A. Proof of the Main Results}
\label{mainproof}

This section provides proofs for the main results, only stating the required lemmas. These lemmas will be proven in Appendix \href{lemmas1}{C}.

\subsection*{A.1 Dual Certificate}

The first lemma establishes dual certificate conditions relative to $T$ for recovery error bounds. It is stated in general form, applicable to any linear operator. For the statement, it is important to notice that for any $X\in\mathbbm{C}^{n_1\times n_2}$ there exists a $G\in T\cap S_{op}$ such that $\|\mathcal{P}_{T}(D)\|_* = \langle D,G\rangle$ by characterization of the nuclear norm via its dual norm (the operator norm). Furthermore, for $\mathcal{A}:\mathbbm{C}^{n_1\times n_2}\mapsto\mathbbm{C}^{m}$ define
\[
\|\mathcal{A}\|_{F\mapsto 2} \coloneqq \max_{X\in S}\|\mathcal{A}(X)\|_2.
\]

\begin{lemma}
\label{dual}
Let $\mathcal{A}:\mathbbm{C}^{n_1\times n_2}\mapsto\mathbbm{C}^{m}$ be a linear operator, and
\begin{equation}
\label{ineq}
\inf_{X\in T\cap S}\|\mathcal{A}(X)\|_2 \geq \beta_1 > 0, \ \ \ \ \|\mathcal{A}\mathcal{P}_{T^{\perp}}\|_{F\mapsto 2} \leq \beta_2.
\end{equation}
Given $D\in\mathbbm{C}^{n_1\times n_2}$ with $\|\mathcal{P}_{T}(D)\|_* = \langle D,G\rangle$ for some $G\in T\cap S_{op}$, assume that there exist $Y, Z\in\mathbbm{C}^{n_1\times n_2}$ with $Y = \mathcal{A}^{*}\mathcal{A}(Z)$ satisfying
\begin{equation}
\label{ineq2}
\|\mathcal{P}_T(Y) - G\|_F \leq \beta_3, \ \ \ \ \|\mathcal{P}_{T^{\perp}}(Y)\| \leq \beta_4, \ \ \ \ \|\mathcal{A}(Z)\|_2\leq\beta_5
\end{equation}
and $\frac{\beta_2\beta_3}{\beta_1} + \beta_4 < 1$. Let $d\in\mathbbm{C}^m$ with $\|d\|_2\leq\eta$ and
\[
D^{\sharp} \coloneqq \argmin\limits_{X\in\mathbbm{C}^{n_1\times n_2}}\|X\|_* \ \ \mbox{s.t.} \ \ \|\mathcal{A}(X)-\mathcal{A}(D)-d\|_2\leq\eta.
\]
Then
\[
\label{NNbd1}
\|D-D^{\sharp}\|_F \leq C_1\|\mathcal{P}_{T^{\perp}}(D)\|_* + C_2\eta,
\]
where $C_1, C_2$ depend on the $\beta_k$'s.

\end{lemma}
The proof of this lemma is postponed until \href{lemmadualproof}{C.1}, where the dependence of $C_1$ and $C_2$ on the $\beta_k$'s is specified. The main result will be obtained by establishing (\ref{ineq}) and (\ref{ineq2}) for the matrix completion sampling operator $P_{\Omega}$. With this in mind, the proof of Theorem \ref{thm2} is provided, from which the remaining theorems are corollaries (see Appendix \href{corollaries}{B}). The proof considers the sampling with replacement model, discussed in detail in the next section.

\subsection*{A.2 Sampling Model}
\label{replacement}

As in previous work, the work load will be simplified by considering the uniform sampling with replacement model. In other words, let $\tilde{\Omega}$ be generated by choosing $m$ entries independently and uniformly at random from $[n_1]\times [n_2]$ (this allows $\tilde{\Omega}$ to have repeated entries). \cite{MC2} and \cite{MC3} show that any upper bound on the probability of failure for exact (noiseless) matrix completion via $\tilde{\Omega}$ is also valid for uniform sampling without replacement. This strategy will apply in the scenario of this work, proceeding in a different manner to include noisy observations, full rank matrices, and inexact subspace estimates.

With $\tilde{\Omega}$ generated as above, let $\Omega\subseteq\tilde{\Omega}$ consist of the $\big|\Omega\big|\leq m$ distinct samples in $\tilde{\Omega}$ and define the normalized operators
\[
\tilde{\mathcal{A}} = \sqrt{\frac{n_1n_2}{m}}P_{\tilde{\Omega}} \ \ \mbox{and} \ \ \mathcal{A} = \sqrt{\frac{n_1n_2}{m}}P_{\Omega},
\]
which extract and scale an input matrix's values in the entries specified by the subset in the subscript. As shown by \cite{MC3} (Section II), the distribution of $\Omega$ is the same as the distribution of sampling $\big|\Omega\big|$ entries uniformly without replacement. Therefore, assume $\Omega\subseteq\tilde{\Omega}$ WLOG, so that generating $\tilde{\Omega}$ with $m$ entries will generate $\Omega$ uniformly at random and any lower bound requisite for $|\Omega|$ will also be satisfied by $m$ (since $m\geq|\Omega|$).

Notice that for any $X\in\mathbbm{C}^{n_1\times n_2}$ and $(k,\ell)\in[n_1]\times [n_2]$
\[
\left(P_{\Omega}^{*}P_{\Omega}(X)\right)_{k\ell} =
\left\{
	\begin{array}{ll}
		X_{k\ell}  & \mbox{if } (k,\ell)\in\Omega \\
		0 & \mbox{otherwise}
	\end{array}
\right.
\]
and
\[
\left(P_{\tilde{\Omega}}^{*}P_{\tilde{\Omega}}(X)\right)_{k\ell} =
\left\{
	\begin{array}{ll}
		\mbox{mult}_{\tilde{\Omega}}(k,\ell)X_{k\ell}  & \mbox{if } (k,\ell)\in\tilde{\Omega} \\
		0 & \mbox{otherwise}
	\end{array}
\right.
\]
where $\mbox{mult}_{\tilde{\Omega}}(k,\ell)\in\mathbbm{N}$ is the multiplicity of $(k,\ell)\in\tilde{\Omega}$. Let $\mbox{mult}_{\tilde{\Omega}}(k,\ell)\leq\tau$ for all $(k,\ell)\in\tilde{\Omega}$. The quantity $\tau$ will be bounded using Proposition 3.3 by \cite{MC2} with parameter $\beta = 3/2$ therein. For $n_1\geq 9$, the proposition gives $\tau\leq 4\log(n_1)$ with probability exceeding $1-n_1^{-1}$. Then for any $X\in\mathbbm{C}^{n_1\times n_2}$
\begin{equation}
\label{boundop}
\|\tilde{\mathcal{A}}(X)\|_2 \leq \sqrt{\tau}\|\mathcal{A}(X)\|_2 \leq 2\sqrt{\log(n_1)}\|\mathcal{A}(X)\|_2.
\end{equation}
The remainder will operate in this scenario so that (\ref{boundop}) holds with high probability. Therefore, using the lower bound
\[
\inf_{X\in T\cap S}\|\tilde{\mathcal{A}}(X)\|_2 \geq \tilde{\beta}_1
\]
for some $\tilde{\beta}_1 >0$, allows to choose parameter $\beta_1 = \frac{\tilde{\beta}_1}{2\sqrt{\log(n_1)}} > 0$ from Lemma \ref{dual}. The following lemma is crucial for this purpose, and will also be used to compute parameters $\beta_3$ and $\beta_5$ with an appropriate choice of dual certificate.

\begin{lemma}
\label{concentration}
Let $T\subset\mathbbm{C}^{n_1\times n_2}$ be a subspace with subspace joint incoherence parameter $M_1$ and $\rho$ defined as in (\ref{dimT}). With $\tilde{\mathcal{A}}$ as above and $0<\delta \leq \frac{1}{4}$, if
\[
\sqrt{m} \geq \frac{C\sqrt{M_1\rho}\log^{3/2}(n_1+n_2)}{\delta}
\]
where $C>0$ is an absolute constant, then
\[
\sup_{X\in T\cap S}\left|\langle(\tilde{\mathcal{A}}^{*}\tilde{\mathcal{A}} - \mathcal{I})(X),X\rangle\right| < 2\delta,
\]
with probability exceeding
\[
1-\exp\left(-\frac{6m\delta^2}{19M_1 \rho}\right).
\]
\end{lemma}
The proof can be found in \href{proofconcentration}{C.2}, which generalizes the approach of \cite{stability,sparseF} (Rudelson-Vershynin Lemma via Dudley's inequality) as done by \cite{pauli}. The following lemma by \cite{optinc} is needed to compute $\beta_4$.

\begin{lemma}
\label{opnorm}
Suppose $Z\in\mathbbm{C}^{n_1\times n_2}$ is a fixed matrix. Then with probability greater than $1-\frac{1}{n_1+n_2}$
\[
\|\tilde{\mathcal{A}}^{*}\tilde{\mathcal{A}}(Z)-Z\|\leq \frac{4n_1n_2\log(n_1+n_2)}{3m}\|Z\|_{\infty} + 2\sqrt{\frac{2n_1n_2\log(n_1+n_2)}{m}}\|Z\|_{\infty,2}.
\]
\end{lemma}
The concentration inequality uses the norm 
\begin{equation}
\label{infty2}
\|Z\|_{\infty,2} \coloneqq \max\Big\{\max_{k}\sqrt{\sum_{j}|Z_{kj}|^2},\max_{\ell}\sqrt{\sum_{j}|Z_{j\ell}|^2}\Big\}
\end{equation}
which is the maximum of the row and column norms of $Z$. The result appears as Lemma 2 in the work of \cite{optinc}, the proof is adapted here to the sampling with replacement model (see \href{extralemmas}{C.3} for the proof). The parameters in (\ref{ineq}) and (\ref{ineq2}) may now be computed to establish the weighted matrix completion error bounds.

\subsection*{A.3 Proof of Theorem \ref{thm2}: Weighted Matrix Completion}
\label{SecWMC}

With respect to the normalized operators, the output of (\ref{wNNmins}) is equivalently given as
\begin{equation}
D^{\omega} \coloneqq \argmin\limits_{X\in\mathbbm{C}^{n_1\times n_2}}\big\|\omega\mathcal{P}_T(X)+\mathcal{P}_{T^{\perp}}(X)\big\|_* \ \ \mbox{s.t.} \ \ \Bigg\|\mathcal{A}(D) + \sqrt{\frac{n_1n_2}{m}}d - \mathcal{A}(X)\Bigg\|_2\leq\sqrt{\frac{n_1n_2}{m}}\eta. \nonumber
\end{equation}
Further note that
\begin{align}
\label{wNNinv}
&\omega\mathcal{P}_T(D^{\omega})+\mathcal{P}_{T^{\perp}}(D^{\omega}) = \nonumber \\
&\argmin\limits_{X\in\mathbbm{C}^{n_1\times n_2}}\|X\|_* \ \ \mbox{s.t.} \ \ \Bigg\|\mathcal{A}(D) + \sqrt{\frac{n_1n_2}{m}}d - \mathcal{A}\left(\frac{1}{\omega}\mathcal{P}_T(X)+\mathcal{P}_{T^{\perp}}(X)\right)\Bigg\|_2\leq\sqrt{\frac{n_1n_2}{m}}\eta,
\end{align}
where equality holds since $\omega\mathcal{P}_T(\circ)+\mathcal{P}_{T^{\perp}}(\circ)$ is invertible when $\omega > 0$, with inverse operator $\omega^{-1}\mathcal{P}_T(\circ)+\mathcal{P}_{T^{\perp}}(\circ)$. Let 
\[
\tilde{\mathcal{P}}_{\omega}(\circ) = \mathcal{P}_T(\circ)+\omega\mathcal{P}_{T^{\perp}}(\circ).
\]
Multiplying both sides of the constraint in (\ref{wNNinv}) by $\omega$ gives
\begin{align}
\label{wNNinv2}
&\omega\mathcal{P}_T(D^{\omega})+\mathcal{P}_{T^{\perp}}(D^{\omega}) = \nonumber \\
&\argmin\limits_{X\in\mathbbm{C}^{n_1\times n_2}}\|X\|_* \ \ \mbox{s.t.} \ \ \Bigg\|\omega\mathcal{A}(D) + \omega\sqrt{\frac{n_1n_2}{m}}d - \mathcal{A}\left(\tilde{\mathcal{P}}_{\omega}(X)\right)\Bigg\|_2\leq\omega\sqrt{\frac{n_1n_2}{m}}\eta\nonumber \\
= &\argmin\limits_{X\in\mathbbm{C}^{n_1\times n_2}}\|X\|_* \ \ \mbox{s.t.} \ \ \Bigg\|\mathcal{A}\left(\tilde{\mathcal{P}}_{\omega}(D_T)\right) + \omega\sqrt{\frac{n_1n_2}{m}}d - \mathcal{A}\left(\tilde{\mathcal{P}}_{\omega}(X)\right)\Bigg\|_2\leq\omega\sqrt{\frac{n_1n_2}{m}}\eta,
\end{align}
were the last line defines $D_T \coloneqq \omega\mathcal{P}_T(D)+\mathcal{P}_{T^{\perp}}(D)$. 

The dual certificate of Lemma \ref{dual} will be produced with respect to program (\ref{wNNinv2}) with sampling operator $\mathcal{A}_T := \mathcal{A}\tilde{\mathcal{P}}_{\omega}$ to obtain an upper bound for 
\[
\big\|D_T-\omega\mathcal{P}_T(D^{\omega})+\mathcal{P}_{T^{\perp}}(D^{\omega})\big\|_F,
\]
which will in turn will give an appropriate bound for $\|D-D^{\omega}\|_F$.

\begin{proof}[Proof of Theorem \ref{thm2}]

Using the notation above, notice that 
\begin{equation}
\label{Abound}
\inf_{X\in T\cap S}\|\mathcal{A}_T(X)\|_2 = \inf_{X\in T\cap S}\|\mathcal{A}(X)\|_2.
\end{equation}
The parameters in (\ref{ineq}) and (\ref{ineq2}) from Theorem \ref{dual} with respect to $\mathcal{A}_T$, $D_T$, $\tilde{U}$ and $\tilde{V}$ will now be bounded. Let $G\in T\cap S_{op}$ be such that $\|\mathcal{P}_T(D_T)\|_* = \langle D_T,G\rangle$. For the dual certificate, choose
\[
Y \coloneqq \mathcal{A}_T^{*}\mathcal{A}_T(Z),
\]
where
\[
Z \coloneqq \tilde{\mathcal{P}}_{\omega}^{-1}\left(P_{\tilde{\Omega}}^{*}P_{\tilde{\Omega}}(G)\right).
\]
Notice that
\[
Y \coloneqq \mathcal{A}_T^{*}\mathcal{A}\left(P_{\tilde{\Omega}}^{*}P_{\tilde{\Omega}}(G)\right).
\]
After bounding the $\beta_k$'s, the proof will finish by applying Lemma \ref{dual}.

\begin{itemize}
\item Parameter $\beta_1$: using (\ref{Abound}) and (\ref{boundop}) will give $\beta_1 = (2\sqrt{2\log(n_1)})^{-1}$. To show this, Theorem \ref{concentration} will be applied in two different ways below with $\delta\leq 1/4$ so that for any $X\in T\cap S$
\begin{equation}
\label{conAt}
\left|\langle(\tilde{\mathcal{A}}^{*}\tilde{\mathcal{A}} - \mathcal{I})(X),X\rangle\right| < 2\delta \leq \frac{1}{2}
\end{equation}
and consequently $\|\tilde{\mathcal{A}}(X)\|_2\geq \sqrt{1-2\delta} \geq 1/\sqrt{2}$ which gives $\beta_1 \geq (2\sqrt{2\log(n_1)})^{-1}$ with high probability.

\item Parameter $\beta_2$: this parameter is $\|\mathcal{A}_T\mathcal{P}_{T^{\perp}}\|_{F\to 2} = \omega\|\mathcal{A}\mathcal{P}_{T^{\perp}}\|_{F\to 2}$. With the chosen dual certificate, use the following calculation
\[
\omega\big\|\mathcal{A}\mathcal{P}_{T^{\perp}}\big\|_{F\to 2} \leq \omega\sqrt{\frac{n_1n_2}{m}}\|P_{\Omega}\|_{F\to 2}\|\mathcal{P}_{T^{\perp}}\|_{F\to F} \leq \omega\sqrt{\frac{n_1n_2}{m}} := \beta_2,
\]
where $\|P_{\Omega}\|_{F\to 2}\leq 1$ holds since $\Omega$ has no repeated entries.

\item Parameter $\beta_3$: this parameter can be computed using (\ref{conAt}). Note that
\[
\mathcal{P}_T(Y) = \mathcal{P}_T\circ\mathcal{A}^{*}\mathcal{A}\left(P_{\tilde{\Omega}}^{*}P_{\tilde{\Omega}}(G)\right) = \mathcal{P}_T\circ\tilde{\mathcal{A}}^{*}\tilde{\mathcal{A}}\left(G\right).
\]
Therefore,
\[
\|\mathcal{P}_T(Y)-G\|_F = \|\mathcal{P}_T\circ(\tilde{\mathcal{A}}^{*}\tilde{\mathcal{A}}-\mathcal{I})\circ\mathcal{P}_T(G)\|_F \leq 2\delta\sqrt{\rho} \coloneqq \beta_3,
\]
since $\|G\| = 1$ gives $\|G\|_F \leq \sqrt{\mbox{rank}(G)} \leq \sqrt{\rho}$ and
\[
\sup_{X\in T\cap S}\left|\langle(\tilde{\mathcal{A}}^{*}\tilde{\mathcal{A}} - \mathcal{I})(X),X\rangle\right| = \|\mathcal{P}_T\circ(\tilde{\mathcal{A}}^{*}\tilde{\mathcal{A}} - \mathcal{I})\circ\mathcal{P}_T\|_{F\mapsto F} \leq 2\delta.
\]
Later, $\delta$ will be chosen in two different ways (always satisfying $\delta\leq 1/4$) which will change the value of $\beta_3$ in each scenario.

\item Parameter $\beta_5$: under the scenario mult$(k,\ell)\leq\tau\leq 4\log(n_1)$ (see discussion in \href{replacement}{A.2}), it can be shown that $\beta_5 \leq \sqrt{6\rho\log(n_1)}$ as follows
\begin{align}
&\Big\|\mathcal{A}(P_{\tilde{\Omega}}^{*}P_{\tilde{\Omega}}(G))\Big\|_2 \coloneqq \sqrt{\frac{n_1n_2}{m}}\Big\|P_{\Omega}\left(P_{\tilde{\Omega}}^{*}P_{\tilde{\Omega}}(G)\right)\Big\|_2 = \sqrt{\frac{n_1n_2}{m}}\Big\|P_{\tilde{\Omega}}^{*}P_{\tilde{\Omega}}(G)\Big\|_F \nonumber \\
\leq &\sqrt{\frac{\tau n_1n_2}{m}}\Big\|P_{\tilde{\Omega}}(\tilde{U}\tilde{\Sigma}\tilde{V}^{\top})\Big\|_2 \coloneqq \sqrt{\tau}\|\tilde{\mathcal{A}}(G)\|_2 \leq \sqrt{6\log(n_1)}\|G\|_F. \nonumber
\end{align}
The last inequality follows from (\ref{conAt}) with $2\delta\leq 1/2$, since $G \in T$.

\end{itemize} 

\noindent The remaining parameter $\beta_4$ is bounded by considering distinct ranges for $\omega$ and choices for $\delta\leq 1/4$ (which will also determine $\beta_3$).\\

\noindent\textbf{Case $\omega \leq \frac{\sqrt{M_1}\log(n_1+n_2)}{\sqrt{n_1n_2}}$:} in this scenario, apply Lemma \ref{concentration} with $\delta = 1/4$, which holds if
\begin{equation}
\label{scom2}
\sqrt{m} \geq 4C\sqrt{M_1\rho}\log^{3/2}(n_1+n_2).
\end{equation}
With appropriate $C$, the probability of success exceeds $1-(n_1+n_2)^{-1}$. Notice that this case gives $\beta_3 = \sqrt{\rho}/2$.

\begin{itemize}

\item Parameter $\beta_4$: to bound $\big\|\mathcal{P}_{T^{\perp}}(Y)\big\|$, notice that
\begin{align}
&\big\|\mathcal{P}_{T^{\perp}}(Y)\big\| = \omega\Big\|\mathcal{P}_{T^{\perp}}\circ\mathcal{A}^{*}\mathcal{A}\left(P_{\tilde{\Omega}}^{*}P_{\tilde{\Omega}}(G)\right)\Big\| = \omega\Big\|\mathcal{P}_{T^{\perp}}\left(\mathcal{A}^{*}\mathcal{A}\left(P_{\tilde{\Omega}}^{*}P_{\tilde{\Omega}}(G)\right) - G\right)\Big\| \nonumber \\ 
&\leq \omega\Big\|\mathcal{A}^{*}\mathcal{A}\left(P_{\tilde{\Omega}}^{*}P_{\tilde{\Omega}}(G)\right) - G\Big\| = \omega\Big\|\tilde{\mathcal{A}}^{*}\tilde{\mathcal{A}}\left(G\right) - G\Big\|, \nonumber
\end{align}

where the inequality holds since $\mathcal{P}_{T^{\perp}}$ is an orthogonal projection. Using the bound above and Lemma \ref{opnorm} gives
\[
\big\|\mathcal{P}_{T^{\perp}}(Y)\big\| \leq \frac{4\omega n_1n_2\log(n_1+n_2)}{3m}\|G\|_{\infty} + 2\omega\sqrt{\frac{2n_1n_2\log(n_1+n_2)}{m}}\|G\|_{\infty,2}
\]
with probability at least $1-(n_1+n_2)^{-1}$. Using the subspace incoherence condition, it is clear that
\begin{equation}
\label{inco2}
\|G\|_{\infty,2} \leq \sqrt{\frac{M_0\rho}{n_2}} \ \ \ \ \mbox{and} \ \ \ \ \|G\|_{\infty} \leq \|G\|_F\sqrt{\frac{M_1\rho}{n_1n_2}} \leq \rho\sqrt{\frac{M_1}{n_1n_2}}.
\end{equation}
Using (\ref{scom2}) and (\ref{inco2}) gives 
\[
\big\|\mathcal{P}_{T^{\perp}}(Y)\big\| \leq \frac{\omega\sqrt{n_1n_2}}{12C^2\sqrt{M_1}\log^2(n_1+n_2)} + \frac{\omega\sqrt{M_0n_1}}{\sqrt{2}C\sqrt{M_1}\log(n_1+n_2)} := \beta_4
\]
and 
\[
\frac{\beta_2\beta_3}{\beta_1} = \frac{\omega\sqrt{2n_1n_2\rho\log(n_1)}}{\sqrt{m}} \leq \frac{\omega\sqrt{n_1n_2}}{2C\sqrt{2M_1}\log(n_1+n_2)}.
\]
Note that $M_0\leq n_2$, and therefore $\frac{\beta_2\beta_3}{\beta_1} + \beta_4 < 1$ if $\omega \leq \frac{\sqrt{M_1}\log(n_1+n_2)}{\sqrt{n_1n_2}}$ and Theorem \ref{dual} may now be applied.
\end{itemize}

\noindent\textbf{Case $\frac{1}{2\sqrt{2\rho\log(n_1+n_2)}}\leq\omega \leq 1$:} apply Lemma \ref{concentration} with $\delta = \sqrt{m}/8\omega n_1\sqrt{2\rho\log(n_1+n_2)}$. Using $m\leq n_1^2$ and the assumption on $\omega$ gives
\[
\delta = \frac{\sqrt{m}}{8\omega n_1\sqrt{2\rho\log(n_1+n_2)}} \leq \frac{1}{8\omega \sqrt{2\rho\log(n_1+n_2)}} \leq \frac{1}{4},
\]
as desired if
\begin{equation}
\label{scom4}
\sqrt{m} \geq C\max\{\sqrt{M_1},M_0\}\sqrt{\rho}\log^{3/2}(n_1+n_2)\frac{8\omega n_1\sqrt{2\rho\log(n_1+n_2)}}{\sqrt{m}}.
\end{equation}
With an appropriate choice of $C$, the probability of success exceeds $1-(n_1+n_2)^{-1}$. This case gives $\beta_3 = \sqrt{m}/4\omega n_1\sqrt{2\log(n_1+n_2)}$.

\begin{itemize}

\item Parameter $\beta_4$: as in the previous range for $\omega$, Lemma \ref{opnorm} along with (\ref{inco2}) and (\ref{scom4}) gives that with high probability
\[
\big\|\mathcal{P}_{T^{\perp}}(Y)\big\| \leq \frac{1}{6\sqrt{2}C\log(n_1+n_2)} + \frac{\sqrt{\omega}}{\sqrt{C\sqrt{2}\log(n_1+n_2)}} \leq \frac{1}{4} := \beta_4, 
\]
where the last inequality holds with $\omega\leq 1$ and an appropriate choice of $C$. Note that $\frac{\beta_2\beta_3}{\beta_1} = \frac{\sqrt{n_2}}{2\sqrt{n_1}}$, and therefore $\frac{\beta_2\beta_3}{\beta_1} + \beta_4 = 3/4$ so that Theorem \ref{dual} can be applied in this case as well.

\end{itemize}

\noindent This concludes the bounds for all $\beta_k$ parameters in (\ref{ineq}) and (\ref{ineq2}) from Theorem \ref{dual}. In both considered weight $\omega$ ranges, Theorem \ref{dual} gives constants (considering only dominating terms, see proof in \href{lemmadualproof}{C.1})
\begin{align}
&C_1 \sim \frac{\beta_2}{\beta_1} + 1 = \omega\sqrt{\frac{n_1n_2\log(n_1)}{m}} + 1, \nonumber \\ 
&C_2 \sim \frac{\beta_2\beta_5}{\beta_1} + \beta_5 = \sqrt{\log(n_1)}\left(\omega\sqrt{\frac{n_1n_2\rho\log(n_1)}{m}} + 1\right), \nonumber
\end{align}
and error bound 
\begin{align}
\label{orthbound}
&\big\|D_T-\omega\mathcal{P}_T(D^{\omega})-\mathcal{P}_{T^{\perp}}(D^{\omega})\big\|_F \leq C_1\big\|\mathcal{P}_{T^{\perp}}(D_T)\big\|_* + \omega\sqrt{\frac{n_1n_2}{m}}C_2\eta \nonumber \\
&= C_1\big\|\mathcal{P}_{T^{\perp}}(D)\big\|_* + \omega\sqrt{\frac{n_1n_2}{m}}C_2\eta.
\end{align}

The desired error term $\|D-D^{\omega}\|_F$ will be bounded in two different ways to introduce $f(\omega)$ in (\ref{errbd2}) as the minimum of both bounds. 

\underline{Bound 1 for $\|D-D^{\omega}\|_F$:} notice that
\begin{align}
&\big\|D_T-\omega\mathcal{P}_T(D^{\omega})-\mathcal{P}_{T^{\perp}}(D^{\omega})\big\|_F^2 = \big\|\mathcal{P}_{T^{\perp}}(D-D^{\omega}) + \omega\mathcal{P}_{T}(D-D^{\omega})\big\|_F^2 \nonumber \\
&= \big\|\mathcal{P}_{T^{\perp}}(D-D^{\omega})\big\|_F^2 + \omega^2\big\|\mathcal{P}_{T}(D-D^{\omega})\big\|_F^2 \nonumber \\
&\geq  \omega^2\|D-D^{\omega}\|_F^2 + \omega^2\|D-D^{\omega}\|_F^2 = \omega^2\|D-D^{\omega}\|_F^2, \nonumber
\end{align}
so for some absolute constant
\begin{align}
\|D-D^{\omega}\|_F &\leq C\left(\sqrt{\frac{n_1n_2\log(n_1)}{m}} + \frac{1}{\omega}\right)\big\|\mathcal{P}_{T^{\perp}}(D)\big\|_* \nonumber \\
&+ C\sqrt{\frac{n_1n_2\log(n_1)}{m}}\left(\omega\sqrt{\frac{n_1n_2\rho\log(n_1)}{m}} + 1\right)\eta. \nonumber
\end{align}

\underline{Bound 2 for $\|D-D^{\omega}\|_F$:} use the derived properties of $\mathcal{A}$ to obtain
\begin{align}
&\|D-D^{\omega}\|_F^2 = \|\mathcal{P}_T(D-D^{\omega})\|_F^2 + \big\|\mathcal{P}_{T^{\perp}}(D-D^{\omega})\big\|_F^2 \nonumber \\
&\leq \frac{1}{\beta_1^2}\|\mathcal{A}\left(\mathcal{P}_T(D-D^{\omega})\right)\|_2^2 + \big\|\mathcal{P}_{T^{\perp}}(D-D^{\omega})\big\|_F^2 \nonumber \\
&\leq \frac{1}{\beta_1^2}\left(\|\mathcal{A}\left(D-D^{\omega}\right)\|_2 + \big\|\mathcal{A}\left(\mathcal{P}_{T^{\perp}}(D-D^{\omega})\right)\big\|_2\right)^2 + \big\|\mathcal{P}_{T^{\perp}}(D-D^{\omega})\big\|_F^2 \nonumber \\
&\leq \frac{1}{\beta_1^2}\left(2\sqrt{\frac{n_1n_2}{m}}\eta + \sqrt{\frac{n_1n_2}{m}}\big\|\mathcal{P}_{T^{\perp}}(D-D^{\omega})\big\|_F\right)^2 + \big\|\mathcal{P}_{T^{\perp}}(D-D^{\omega})\big\|_F^2, \nonumber
\end{align}
where the last inequality holds by feasibility of $D^{\omega}$ for (\ref{wNNmins}) and since $\mathcal{A} := \sqrt{\frac{n_1n_2}{m}}P_{\Omega}$ contains no repeated entries. The term in the final line can be bounded by (\ref{orthbound}) since
\begin{align}
&\big\|\mathcal{P}_{T^{\perp}}(D-D^{\omega})\big\|_F^2 \leq \big\|\mathcal{P}_{T^{\perp}}(D-D^{\omega})\big\|_F^2 + \omega\|\mathcal{P}_{T}(D-D^{\omega})\|_F^2 \nonumber \\
&= \big\|\mathcal{P}_{T^{\perp}}(D-D^{\omega}) + \omega\mathcal{P}_{T}(D-D^{\omega})\big\|_F^2 = \big\|D_T-\omega\mathcal{P}_T(D^{\omega})-\mathcal{P}_{T^{\perp}}(D^{\omega})\big\|_F^2. \nonumber
\end{align}

This approach gives that for some absolute constant
\begin{align}
\|D-D^{\omega}\|_F &\leq C\sqrt{\frac{n_1n_2\log(n_1)}{m}}\left(\omega\sqrt{\frac{n_1n_2\log(n_1)}{m}} + 1\right)\big\|\mathcal{P}_{T^{\perp}}(D)\big\|_* \nonumber \\
&+ C\omega\frac{n_1n_2\log(n_1)}{m}\left(\omega\sqrt{\frac{n_1n_2\rho\log(n_1)}{m}} + 1\right)\eta + \sqrt{\frac{n_1n_2\log(n_1)}{m}}\eta. \nonumber
\end{align}

Notice that, in contrast to bound 1 for $\|D-D^{\omega}\|_F$, bound 2 includes the multiplicative term $\omega\sqrt{n_1n_2\log(n_1)/m}$ and the additive term $\sqrt{n_1n_2\log(n_1)/m}$ in the noise term. Factoring out the multiplicative term and choosing the minimum defines $f(\omega)$.

\textcolor{black}{To finish, bound $\big\|\mathcal{P}_{T^{\perp}}(D)\big\|_*$ as follows}
\begin{align}
&\textcolor{black}{\big\|\mathcal{P}_{T^{\perp}}(D)\big\|_* \leq \big\|\mathcal{P}_{T^{\perp}}(U^{\rho}\Sigma^{\rho} V^{\rho \top})\big\|_* + \big\|\mathcal{P}_{T^{\perp}}(U^+\Sigma^+ V^{+\top})\big\|_* \nonumber} \\
&\textcolor{black}{\leq \big\|\mathcal{P}_{T^{\perp}}(U^{\rho}\Sigma^{\rho} V^{\rho \top})\big\|_* + \|U^+\Sigma^+ V^{+\top}\|_* = \big\|\mathcal{P}_{T^{\perp}}(U^{\rho}\Sigma^{\rho} V^{\rho \top})\big\|_* + \sum_{k=\rho+1}^{n_2}\sigma_k(D) \nonumber .}
\end{align}
\end{proof}
\textcolor{black}{To establish Corollary \ref{corthm2} from the proof above, it suffices to bound $\big\|\mathcal{P}_{T^{\perp}}(U^{\rho}\Sigma^{\rho} V^{\rho \top})\big\|_*$ as follows:
\[
\big\|\mathcal{P}_{T^{\perp}}\mathcal{P}_{T_{\rho}}(U^{\rho}\Sigma^{\rho} V^{\rho \top})\big\|_* \leq \sqrt{n_2}\big\|\mathcal{P}_{T^{\perp}}\mathcal{P}_{T_{\rho}}(U^{\rho}\Sigma^{\rho} V^{\rho \top})\big\|_F \leq \sqrt{n_2}\big\|\mathcal{P}_{T^{\perp}}\mathcal{P}_{T_{\rho}}\big\|_{F\rightarrow F}\|U^{\rho}\Sigma^{\rho} V^{\rho \top}\|_F.
\]}

\section*{Appendix B. Proof of Theorems \ref{thm1}, \ref{cor1}, \ref{cor2}, \ref{cor4}, and \ref{cororiginal}}
\label{corollaries}

Theorems \ref{cor1}, \ref{cor2}, \ref{cor4}, and \ref{cororiginal} are all corollaries of Theorem \ref{thm2} with small $\omega$ while Theorem \ref{thm1} applies $\omega = 1$. To obtain these results, it is sufficient to bound the $\rho$-subspace incoherence parameters for the respective estimate subspaces (all with $\rho = r$) and the term \textcolor{black}{$\big\|\mathcal{P}_{T^{\perp}}(U^{r}\Sigma^{r} V^{r \top})\big\|_*$} in \eqref{errbd2}. 

For Theorem \ref{thm1}, notice that with $\omega = 1$ program (\ref{wNNmins}) becomes the nuclear norm minimization program (\ref{NNmin}). This makes $T$ irrelevant, allowing the choice of subspace as in Theorem \ref{cor1} with $\tilde{U}=U^r$ and $\tilde{V}=V^r$ to provide \textcolor{black}{$\big\|\mathcal{P}_{T^{\perp}}(U^{r}\Sigma^{r} V^{r \top})\big\|_*=0$}. It only remains to upper bound $M_0$ in (\ref{Mu1}) and $M_1$ (\ref{Mu2}) in this case.

\begin{itemize}
\item \underline{Theorem \ref{cor2}, program (\ref{wNNmins}, $T_2$):} applying Theorem \ref{thm2} with small weights requires upper and lower bounding the subspace joint incoherence parameter $M_1$. Recall that any $X\in T_2\cap S$ satisfies $X = \tilde{U}\tilde{U}^{\top}X\tilde{V}\tilde{V}^{\top}$, so $X$ can be written as $X = \tilde{U}W^{\top}$ and $X = Z\tilde{V}^{\top}$ where range$(W)\subset$ range$(\tilde{V})$ and range$(Z)\subset$ range$(\tilde{U})$.

To upper bound the subspace joint incoherence, write $X = \tilde{U}W^{\top}$ and notice that $W = \tilde{V}\alpha$ where $\alpha\in\mathbbm{C}^{r\times r}$ and $\|\alpha\|_F = 1$. Then,
\begin{align}
&\Big|(\tilde{U}W^{\top})_{pq}\Big| = \Bigg|\sum_{k}\tilde{U}_{pk}\overline{W}_{qk}\Bigg| = \Bigg|\sum_{kj}\tilde{U}_{pk}\overline{\alpha}_{jk}\overline{\tilde{V}}_{qj}\Bigg| \nonumber \\
&\leq \|\alpha\|_F\left(\sum_{kj}|\tilde{U}_{pk}|^2|\tilde{V}_{qj}|^2\right)^{1/2} \leq \frac{\sqrt{\mu_0(\tilde{U}\tilde{U}^{\top})\mu_0(\tilde{V}\tilde{V}^{\top})}r}{\sqrt{n_1n_2}}. \nonumber
\end{align}
Therefore $M_1(T_2) \leq \mu_0(\tilde{U}\tilde{U}^{\top})\mu_0(\tilde{V}\tilde{V}^{\top})r$. 

The lower bound will be obtained by a proper selection of $\alpha$. Let $\tilde{p}\in [n_1]$ obtain the maximum row norm of $\tilde{U}$ and likewise $\tilde{q}\in [n_2]$ for $\tilde{V}$. Then with $\alpha_{jk} = c\tilde{U}_{\tilde{p}k}\overline{\tilde{V}}_{\tilde{q}j}$ where $c$ is a normalization constant (so that $\|\alpha\|_F=1$) gives
\[
\Big|(\tilde{U}W^{\top})_{\tilde{p}\tilde{q}}\Big| = \|\tilde{U}_{\tilde{p}*}\|_2\|\tilde{V}_{\tilde{q}*}\|_2 = \frac{r\sqrt{\mu_0(\tilde{U}\tilde{U}^{\top})\mu_0(\tilde{V}\tilde{V}^{\top})}}{\sqrt{n_1n_2}}
\]
to obtain $M_1(T_2) = \mu_0(\tilde{U}\tilde{U}^{\top})\mu_0(\tilde{V}\tilde{V}^{\top})r$.

\textcolor{black}{Finally, to bound $\big\|\mathcal{P}_{T_2^{\perp}}(X)\big\|_*$ with $X = U^{r}\Sigma^{r} V^{r \top}$ gives
\begin{align*}
    &\big\|\mathcal{P}_{T_2^{\perp}}(X)\big\|_* = \|\tilde{U}\tilde{U}^{\top}X\tilde{V}^{\perp}\tilde{V}^{\perp\top} + \tilde{U}^{\perp}\tilde{U}^{\perp\top}X\|_* \\
    &\leq \sqrt{r}\|\tilde{U}\tilde{U}^{\top}X\tilde{V}^{\perp}\tilde{V}^{\perp\top}\|_F + \sqrt{r}\|\tilde{U}^{\perp}\tilde{U}^{\perp\top}X\|_F \leq \sqrt{r}(\sin(\theta_v)+\sin(\theta_u))\|X\|_F.
\end{align*}}

\item \underline{Theorem \ref{cor4}, program (\ref{wNNmins}, $T_3$):} notice that any $X\in T_3\cap S$ can be written as $X = \tilde{U}W^{\top}$ where the columns of $W$ are arbitrary. Then
\[
\Big|(\tilde{U}W^{\top})_{pq}\Big|  \leq \|\tilde{U}_{p*}\|_2\|W_{q*}\|_2 \leq \sqrt{\frac{\mu_0(\tilde{U}\tilde{U}^{\top})r}{n_1}},
\]
and therefore, $M_1 \leq \mu_0(\tilde{U}\tilde{U}^{\top})n_2$. For a lower bound, let $\tilde{p}\in [n_1]$ obtain the maximum row norm of $\tilde{U}$ and choose each row $W_{q*} = c\tilde{U}_{\tilde{p}*}$ were $c$ is a proper normalizing constant achieving $\|W\|_F=1$. Then
\[
M_1 \geq \frac{n_1n_2}{r}\Big|(\tilde{U}W^{\top})_{\tilde{p}q}\Big|^2 = \frac{n_2\|\tilde{U}_{\tilde{p}*}\|_2^2}{r} = \frac{\mu_0(\tilde{U}\tilde{U}^{\top})n_2}{n_1}.
\]

\textcolor{black}{$\big\|\mathcal{P}_{T_3^{\perp}}(U^{r}\Sigma^{r} V^{r \top})\big\|_*$ can be bounded as 
\begin{align*}
    &\big\|\mathcal{P}_{T_3^{\perp}}(U^{r}\Sigma^{r} V^{r \top})\big\|_* = \|\tilde{U}^{\perp}\tilde{U}^{\perp\top}U^{r}\Sigma^{r} V^{r \top}\|_* \leq \sqrt{r}\|\tilde{U}^{\perp}\tilde{U}^{\perp\top}U^{r}\Sigma^{r} V^{r \top}\|_F \\
    &\leq \sqrt{r}\sin(\theta_u)\|U^{r}\Sigma^{r} V^{r \top}\|_F.
\end{align*}}

\item \underline{Theorem \ref{cororiginal}, program (\ref{wNNmins}, $T_4$):} any $X\in T_4\cap S_{op}$ can be written as $X = X_1 + X_2 + X_3$ where $X_1\in T_2$, $X_2 = \tilde{U}W^{\top}$ with range$(W)\subset$ range$(\tilde{V}^{\perp})$, and $X_3 = Z\tilde{V}^{\top}$ with range$(Z)\subset$ range$(\tilde{U}^{\perp})$. As before, the largest entry of any $X_1$ is bounded by $\mu_0r/\sqrt{n_1n_2}$. For matrices of the form $X_3$, write $Z = \tilde{U}^{\perp}\alpha$ where $\alpha\in\mathbbm{C}^{n_1-r\times r}$ has orthogonal columns and $\|\alpha\|_F\leq 1$. Then,
\[
\Big|(Z\tilde{V}^{\top})_{pq}\Big|  \leq \|Z_{p*}\|_2\|\tilde{V}_{q*}\|_2 \leq \|\alpha\|\|\tilde{U}_{p*}^{\perp}\|_2\sqrt{\frac{\mu_0 r}{n_2}} \leq \sqrt{\frac{\mu_0 r(n_1-\mu_2r)}{n_1n_2}}.
\]
The last inequality holds since $1 = \|\tilde{U}_{p*}\|_2 + \|\tilde{U}_{p*}^{\perp}\|_2^2$ and by definition (\ref{munew}). An analogous argument for $X_2$ and the triangle inequality gives 
\[
M_1\leq \mu_0^2 r + \mu_0 (n_1-\mu_2r) + \mu_0 (n_2-\mu_2r).
\]
For the lower bound, consider matrices of the form $X_3$ as before. Choose $\alpha$ with $\|\alpha\|_F=1$ in an analogous manner to the proof of Theorem \ref{cor2}. Then $M_1 \geq \mu_0(\tilde{V}\tilde{V}^{\top})(n_1-\mu_2r)\geq n_1-\mu_2r$. 

\textcolor{black}{For $\big\|\mathcal{P}_{T_4^{\perp}}(U^{r}\Sigma^{r} V^{r \top})\big\|_*$, it follows that
\begin{align}
&\big\|\mathcal{P}_{T_4^{\perp}}(U^{r}\Sigma^{r} V^{r \top})\big\|_* = \|\tilde{U}^{\perp}\tilde{U}^{\perp\top}U^{r}\Sigma^{r} V^{r \top}\tilde{V}^{\perp}\tilde{V}^{\perp\top}\|_*\\
&\leq \sqrt{r}\|\tilde{U}^{\perp}\tilde{U}^{\perp\top}U^{r}\Sigma^{r} V^{r \top}\tilde{V}^{\perp}\tilde{V}^{\perp\top}\|_F \leq \sqrt{r}\sin(\theta_u)\sin(\theta_v) \|U^{r}\Sigma^{r} V^{r \top}\|_F. \nonumber
\end{align}}

\item \underline{Theorem \ref{cor1}, program (\ref{wNNmins}, $T_1$):} recall that $T_1 =$ span$\{\tilde{U}_{k*}\tilde{V}_{k*}^{\top}\}_{k\in [r]}$, so every $X\in T_1\cap S$ can be written as $X = \tilde{U}\Sigma\tilde{V}^{\top}$ for some diagonal matrix with $\|\Sigma\|_F\leq 1$.

For the upper bound, it holds that for any $p\in [n_1]$ and $q\in[n_2]$ 
\[
\Big|(\tilde{U}\Sigma\tilde{V}^{\top})_{pq}\Big| = \Bigg|\sum_{k}\sigma_{k}\tilde{U}_{pk}\overline{\tilde{V}}_{qk}\Bigg| \leq \left(\sum_k\sigma_k^2\right)^{1/2}\left(\sum_{k}|\tilde{U}_{pk}|^2|\tilde{V}_{qk}|^2\right)^{1/2} \leq \sqrt{\frac{\mu_1(\tilde{U}\tilde{V}^{\top})r}{n_1n_2}}.
\]
This gives $M_1(T_1) \leq \mu_1(\tilde{U}\tilde{V}^{\top})$.

For the lower bound, notice that the singular values can be freely chosen. Let $(\tilde{p},\tilde{q})\in [n_1]\times [n_2]$ be such that
\[
\left(\sum_{k}|\tilde{U}_{\tilde{p}k}|^2|\tilde{V}_{\tilde{q}k}|^2\right)^{1/2} = \sqrt{\frac{\mu_1(\tilde{U}\tilde{V}^{\top})r}{n_1n_2}}.
\]
Then choosing $\sigma_k = c\overline{\tilde{U}}_{\tilde{p}k}\tilde{V}_{\tilde{q}k}$ (possibly complex valued which will still be in $T_1$) where $c$ is a normalization constant (so $\|\Sigma\|_F=1$) gives
\[
\Big|(\tilde{U}\Sigma\tilde{V}^{\top})_{\tilde{p}\tilde{q}}\Big| = \sqrt{\frac{\mu_1(\tilde{U}\tilde{V}^{\top})r}{n_1n_2}},
\]
and therefore $M_1(T_1) \geq \mu_1(\tilde{U}\tilde{V}^{\top})$ since it is defined as the maximum.

\textcolor{black}{Finally, to bound $\big\|\mathcal{P}_{T_1^{\perp}}(U^{r}\Sigma^{r} V^{r \top})\big\|_*$ with $X = U^{r}\Sigma^{r} V^{r \top}$ notice that
\[
\mathcal{P}_{T_1^{\perp}}(X) = \mathcal{P}_{T_2^{\perp}}(X) + \sum_{k\neq\ell} \tilde{U}_{*k}\tilde{V}_{*\ell}^{\top}\langle\tilde{U}_{*k}\tilde{V}_{*\ell}^{\top},X\rangle,
\]
where $T_2$ is as in Theorem \ref{cor2} and the term $\|\mathcal{P}_{T_2^{\perp}}(X)\|_*$ can be bounded as in the proof therein via a triangle inequality. To bound the remaining term, which is a rank $r$ matrix, we obtain
\begin{align*}
    &\Bigg\|\sum_{k\neq\ell} \tilde{U}_{*k}\tilde{V}_{*\ell}^{\top}\langle\tilde{U}_{*k}\tilde{V}_{*\ell}^{\top},X\rangle\Bigg\|_* \leq \sqrt{r}\Bigg\|\sum_{k\neq\ell} \tilde{U}_{*k}\tilde{V}_{*\ell}^{\top}\langle\tilde{U}_{*k}\tilde{V}_{*\ell}^{\top},X\rangle\Bigg\|_F \\
    &= \sqrt{r}\left(\sum_{k\neq\ell}\lvert \langle\tilde{U}_{*k}\tilde{V}_{*\ell}^{\top},X\rangle\rvert^2\right)^{1/2} = \sqrt{r}\left(\sum_{k\neq\ell}\lvert \langle U^{r\top}\tilde{U}_{*k}\tilde{V}_{*\ell}^{\top}V^{r},\Sigma^r\rangle\rvert^2\right)^{1/2}\\
    &\leq \sqrt{r}\|\Sigma^r\|_F\left(\sum_{k\neq\ell}\big\| U^{r\top}\tilde{U}_{*k}\tilde{V}_{*\ell}^{\top}V^{r}\big\|_F^2\right)^{1/2}.
\end{align*}}

\item \underline{Theorem \ref{thm1}, program \eqref{NNmin}:} as discussed, choose $T$ as $T_1$ from Theorem \ref{cor1} but with $\tilde{U}=U^r$ and $\tilde{V}=V^r$. Then \textcolor{black}{$\big\|\mathcal{P}_{T_1^{\perp}}(U^{r}\Sigma^{r} V^{r \top})\big\|_*= 0$} and as in the previous case $M_1 = \mu_1(D)$.

To upper bound $M_0$, every $X\in T\cap S_{op}$ can be written as $X = U^r\Sigma V^{r\top}$ for some diagonal matrix with $\|\Sigma\|\leq 1$. Then $M_0(T) \leq \mu_0(D)$, since
\begin{align}
&\|X\|_{2,\infty} = \max_{k,\ell}\{\|X_{k*}\|_2,\|X_{*\ell}\|_2\} = \max_{k,\ell}\{\|U_{k*}^r\Sigma V^{r\top}\|_2,\|U^r\Sigma V_{\ell *}^{r\top}\|_2\} \nonumber \\ 
&\leq \max_{k,\ell}\{\|U_{k*}^r\|_2\|\Sigma V^{r\top}\|,\|U^r\Sigma\|\|V_{\ell *}^r\|_2\} \leq \max_{k,\ell}\{\|U_{k*}^r\|_2,\|V_{\ell *}^r\|_2\} \leq \sqrt{\frac{\mu_0(D)r}{n_2}}, \nonumber
\end{align}
where the second inequality holds since $U^r, V^r$ with orthonormal columns and $\|\Sigma\|\leq 1$ give that $\|U^r\Sigma\|$ and $\|\Sigma V^{r\top}\|$ are bounded by 1.

\end{itemize}

\section*{Appendix C. Proof of Required Lemmas}
\label{lemmas1}

This section proves the main lemmas required for the proofs in Appendix \href{mainproof}{A}: Lemma \ref{dual} and Lemma \ref{concentration}. 

\subsection*{C.1 Proof of Lemma \ref{dual}}
\label{lemmadualproof}

Lemma \ref{dual} is essentially a generalization of dual certificate guarantees for sparse vector recovery to the low-rank matrix recovery case (see Theorem 4.33 by \cite{introCS}), and the proof will be similar.

\begin{proof}[Proof of Lemma \ref{dual}]
Denote $W = D^{\sharp} - D $, the goal is to bound $\|W\|_F$. Since $D^{\sharp}$ is feasible
\[
\|\mathcal{A}(W)\|_2 \leq \|\mathcal{A}(D^{\sharp})-\mathcal{A}(D)-d\|_2 + \|d\|_2 \leq 2\eta.
\]
This will be applied throughout the proof.

Let $Q\in T^{\perp}$ be such that $\langle D + W, Q\rangle = \|\mathcal{P}_{T^{\perp}}(D + W)\|_*$ and $G\in T$ be such that $\langle D, G\rangle = \|\mathcal{P}_T(D)\|_*$. By optimality of $D^{\sharp}$ and feasibility of $D$,
\begin{align}
&\|D\|_* \geq \|D^{\sharp}\|_* = \|D + W\|_* \geq |\langle D + W, G + Q\rangle| \nonumber \\
= &\big|\langle D + W, G\rangle + \big\|\mathcal{P}_{T^{\perp}}(D + W)\big\|_*\big| \nonumber \\ 
= &\big|\|\mathcal{P}_T(D)\|_* + \langle \mathcal{P}_T(W), G\rangle + \big\|\mathcal{P}_{T^{\perp}}(D + W)\big\|_*\big| \nonumber \\
\geq &\|\mathcal{P}_T(D)\|_* - |\langle \mathcal{P}_T(W), G\rangle| + \big\|\mathcal{P}_{T^{\perp}}(W)\big\|_* - \big\|\mathcal{P}_{T^{\perp}}(D)\big\|_*. \nonumber
\end{align}
Where the second inequality holds by the variational characterization of the nuclear norm, $\|X\|_* = \sup_{\|Y\|\leq 1}\langle X,Y\rangle$.

Using $\|D\|_* \leq \|\mathcal{P}_T(D)\|_*+\|\mathcal{P}_{T^{\perp}}(D)\|_*$ and rearranging gives
\begin{equation}
\label{eq1}
\|\mathcal{P}_{T^{\perp}}(W)\|_* \leq 2\|\mathcal{P}_{T^{\perp}}(D)\|_* + |\langle \mathcal{P}_T(W), G\rangle|.
\end{equation}

Introducing $Y = \mathcal{A}^{*}\mathcal{A}(Z)$, the last term in (\ref{eq1}) can be bounded as
\begin{align}
\label{boundsLR}
&|\langle \mathcal{P}_T(W), G\rangle| \leq |\langle \mathcal{P}_T(W), G - \mathcal{P}_{T}(Y)\rangle| + |\langle \mathcal{P}_T(W), \mathcal{P}_{T}(Y)\rangle| \nonumber \\
\leq &\beta_3\|\mathcal{P}_T(W)\|_F + |\langle \mathcal{P}_T(W), \mathcal{P}_{T}(Y)\rangle| \nonumber \\
= &\beta_3\|\mathcal{P}_T(W)\|_F + |\langle W-\mathcal{P}_{T^{\perp}}(W), Y-\mathcal{P}_{T^{\perp}}(Y)\rangle| \nonumber \\
= &\beta_3\|\mathcal{P}_T(W)\|_F + |\langle W, Y\rangle - \langle \mathcal{P}_{T^{\perp}}(W), \mathcal{P}_{T^{\perp}}(Y)\rangle| \nonumber \\
\leq &\beta_3\|\mathcal{P}_T(W)\|_F + |\langle W, Y\rangle| + |\langle \mathcal{P}_{T^{\perp}}(W), \mathcal{P}_{T^{\perp}}(Y)\rangle|.
\end{align}
The second equality applies $\langle W, \mathcal{P}_{T^{\perp}}(Y)\rangle = \langle \mathcal{P}_{T^{\perp}}(W), \mathcal{P}_{T^{\perp}}(Y)\rangle = \langle \mathcal{P}_{T^{\perp}}(W), Y\rangle$.

The three terms in (\ref{boundsLR}) are bounded next, beginning with $\|\mathcal{P}_T(W)\|_F$. The assumed inequalities (\ref{ineq}) give that $\|\mathcal{A}(B)\|_2 \geq \beta_1\|B\|_F$, and $\|\mathcal{A}(H)\|_2 \leq \beta_2\|H\|_F$ for any $B\in T$ and $H\in T^{\perp}$. Therefore,
\begin{align}
\label{PTW}
\begin{split}
&\|\mathcal{P}_T(W)\|_F \leq \frac{1}{\beta_1}\|\mathcal{A}(\mathcal{P}_T(W))\|_2 \leq \frac{1}{\beta_1}\|\mathcal{A}(W)\|_2 + \frac{1}{\beta_1}\big\|\mathcal{A}(\mathcal{P}_{T^{\perp}}(W))\big\|_2\\
\leq &\frac{2\eta}{\beta_1} + \frac{\beta_2}{\beta_1}\|\mathcal{P}_{T^{\perp}}(W)\|_F. \\
\end{split}
\end{align}

The remaining terms in (\ref{boundsLR}), $|\langle \mathcal{P}_{T^{\perp}}(W), \mathcal{P}_{T^{\perp}}(Y)\rangle|$ and $|\langle W, Y\rangle|$, can be bounded by assumptions (\ref{ineq2})
\[
|\langle \mathcal{P}_{T^{\perp}}(W), \mathcal{P}_{T^{\perp}}(Y)\rangle| \leq \|\mathcal{P}_{T^{\perp}}(W)\|_*\|\mathcal{P}_{T^{\perp}}(Y)\| \leq \beta_4\|\mathcal{P}_{T^{\perp}}(W)\|_*
\]
and
\[
|\langle W, Y\rangle| \coloneqq |\langle W, \mathcal{A}^{*}\mathcal{A}(Z)\rangle| = |\langle \mathcal{A} (W), \mathcal{A}(Z)\rangle| \leq \|\mathcal{A} (W)\|_2 \|\mathcal{A}(Z)\|_2 \leq 2\eta \beta_5.
\]

Using these inequalities to bound $|\langle \mathcal{P}_T(W), G\rangle|$ in (\ref{eq1}) gives
\begin{align}
&\|\mathcal{P}_{T^{\perp}}(W)\|_* \leq 2\|\mathcal{P}_{T^{\perp}}(D)\|_* + \frac{2\eta\beta_3}{\beta_1} + \frac{\beta_2\beta_3}{\beta_1}\|\mathcal{P}_{T^{\perp}}(W)\|_F + \beta_4\|\mathcal{P}_{T^{\perp}}(W)\|_* + 2\eta\beta_5 \nonumber \\
\leq & 2\|\mathcal{P}_{T^{\perp}}(D)\|_* + \left(\frac{\beta_2\beta_3}{\beta_1} + \beta_4\right)\|\mathcal{P}_{T^{\perp}}(W)\|_* + 2\left(\frac{\beta_3}{\beta_1} + \beta_5\right)\eta. \nonumber
\end{align}

Since by assumption $\rho \coloneqq \frac{\beta_2\beta_3}{\beta_1} + \beta_4 < 1$, rearrange to obtain
\[
\|\mathcal{P}_{T^{\perp}}(W)\|_* \leq \frac{2}{1-\rho}\|\mathcal{P}_{T^{\perp}}(D)\|_* + \frac{2\left(\frac{\beta_3}{\beta_1} + \beta_5\right)\eta}{1-\rho}.
\]
From previous calculations, (\ref{PTW}),
\[
\|\mathcal{P}_T(W)\|_F \leq \frac{2\eta}{\beta_1} + \frac{\beta_2}{\beta_1}\|\mathcal{P}_{T^{\perp}}(W)\|_F \leq \frac{2\eta}{\beta_1} + \frac{\beta_2}{\beta_1}\|\mathcal{P}_{T^{\perp}}(W)\|_*,
\]
so that both of these inequalities give
\begin{align}
&\|W\|_F \leq \|\mathcal{P}_T(W)\|_F + \|\mathcal{P}_{T^{\perp}}(W)\|_F \leq \|\mathcal{P}_T(W)\|_F + \|\mathcal{P}_{T^{\perp}}(W)\|_* \nonumber \\ 
\leq  &\frac{2\eta}{\beta_1} + \left(\frac{\beta_2}{\beta_1}+1\right)\|\mathcal{P}_{T^{\perp}}(W)\|_*
\leq C_1\|\mathcal{P}_{T^{\perp}}(D)\|_* + 2C_2\eta, \nonumber
\end{align}
with constants given as
\[
C_1 \coloneqq 2\left(\frac{\beta_2}{\beta_1} + 1\right)\left(1-\frac{\beta_2\beta_3}{\beta_1}-\beta_4\right)^{-1},
\]
and
\[
C_2 \coloneqq \frac{1}{\beta_1} + \left(\frac{\beta_2}{\beta_1} + 1\right)\left(\frac{\beta_3}{\beta_1}+\beta_5\right)\left(1-\frac{\beta_2\beta_3}{\beta_1}-\beta_4\right)^{-1}.
\]

\end{proof}

\subsection*{C.2 Proof of Lemma \ref{concentration}}
\label{proofconcentration}

The proof of Lemma \ref{concentration} requires two lemmas, stated here and proven in Section \href{extralemmas}{C.3}. Before continuing, some useful notation and observations are established for the sampling operators. Recall that $M_1$ is the $\rho$-subspace incoherence parameter of $T$ and that the operator that samples with replacement has been normalized as
\[
\tilde{\mathcal{A}} \coloneqq \sqrt{\frac{n_1n_2}{m}}\mathcal{P}_{\tilde{\Omega}},
\]
where $m$ and $\tilde{\Omega}$ (with possible repetitions) are defined as in Section \href{replacement}{A.2}.

The operator $\tilde{\mathcal{A}}$ is as an ensemble of matrices $\{\tilde{\mathcal{A}}^{k}\}_{k\in [m]}\subset\mathbbm{C}^{n_1\times n_2}$. Here the superscripts order the matrices such that the action on $X\in\mathbbm{C}^{n_1\times n_2}$ is given entry-wise as
\[
\tilde{\mathcal{A}}(X)_{k} = \langle \tilde{\mathcal{A}}^{k},X\rangle,
\]
for $k\in[m]$. The scaling ensures that $\tilde{\mathcal{A}}^{*}\tilde{\mathcal{A}}$ forms an isotropic ensemble, that is, for any $X\in\mathbbm{C}^{n_1\times n_2}$
\begin{equation}
\label{isotropy}
\mathbbm{E}\tilde{\mathcal{A}}^{*}\tilde{\mathcal{A}}(X) = \sum_{k=1}^{m}\mathbbm{E}\tilde{\mathcal{A}}^k\langle \tilde{\mathcal{A}}^{k},X\rangle = \sum_{k=1}^{m}\left(\frac{1}{n_1n_2}\sum_{p=1}^{n_1}\sum_{q=1}^{n_2}\frac{n_1n_2}{m}M^{p,q}X_{pq}\right) = X,
\end{equation}
where $\{M^{p,q}\}_{(p,q)\in[n_1\times n_2]}$ is the canonical $n_1\times n_2$ matrix basis. Therefore, each $\tilde{\mathcal{A}}^k$ is a random matrix that achieves $\sqrt{\frac{n_1n_2}{m}}M^{p,q}$ with probability $(n_1n_2)^{-1}$.

Next is a lemma that will be useful to establish Lemma \ref{concentration}, in order to apply a concentration inequality. The proof is postponed until Section \href{extralemmas}{C.3} and is straightforward from the subspace incoherence assumptions.

\begin{lemma}
\label{lemma1LR2}
Define $\tilde{\mathcal{A}}$ as above. Then for $Z\in T$ and all $k\in [m]$
\begin{equation}
\label{lemma1LReq2}
|\langle\tilde{\mathcal{A}}^k,Z\rangle| \leq \|Z\|_F\sqrt{\frac{M_1\rho}{m}}
\end{equation}
and
\[
\mathbbm{E}\sum_{k=1}^{m}|\langle\tilde{\mathcal{A}}^k,Z\rangle|^4 \leq \|Z\|_{F}^4\frac{M_1 \rho}{m},
\]
where $\rho$ is defined as in (\ref{dimT}).
\end{lemma}

The next Lemma can be considered a generalization of Lemma 3.6 by \cite{sparseF}. The argument is due to \cite{pauli}, but has been tailored to the current setting with a tighter bound in terms of the logarithm degree. Adopting the notation therein, for a matrix $A\in\mathbbm{C}^{n_1\times n_2}$ denote $|A)(A|$ as the operator that maps $X\mapsto A\langle A,X\rangle$.

\begin{lemma}
\label{radem}
Let $m \leq n_1n_2$ and $\epsilon_1, ..., \epsilon_m$ be i.i.d. Rademacher random variables. Then
\begin{align}
&\mathbbm{E}_{\epsilon}\sup_{X\in T\cap S} \sum_{k=1}^{m}\epsilon_k \langle|\tilde{\mathcal{A}}^k)(\tilde{\mathcal{A}}^k|(X),X\rangle \nonumber \\ \leq & \frac{\tilde{C}\sqrt{M_1\rho}\log(n_1+n_2)\log^{1/2}(m)}{\sqrt{m}}\left(\sup_{X\in T\cap S} \sum_{k=1}^{m} \langle|\tilde{\mathcal{A}}^k)(\tilde{\mathcal{A}}^k|(X),X\rangle\right)^{1/2}, \nonumber
\end{align}
where $\tilde{C}>0$ is an absolute constant.
\end{lemma}
See Section \href{extralemmas}{C.3} for the proof. The proof of Lemma \ref{concentration} follows.

\begin{proof}[Proof of Lemma \ref{concentration}]
Let $\mathcal{T} = T\cap S$ and
\[
\mathcal{X} \coloneqq \sup_{X\in \mathcal{T}}\left|\langle(\tilde{\mathcal{A}}^{*}\tilde{\mathcal{A}} - \mathcal{I})(X),X\rangle\right| = \sup_{X\in \mathcal{T}}\langle(\tilde{\mathcal{A}}^{*}\tilde{\mathcal{A}} - \mathcal{I})(X),X\rangle,
\] 
where the last equality holds since $\tilde{\mathcal{A}}^{*}\tilde{\mathcal{A}} - \mathcal{I}$ is a Hermitian operator.

The goal is to show
\[
\mathcal{X} \leq 2\delta,
\]
which is achieved proceeding along the lines of \cite{pauli,stability,sparseF}. Adopting the notation from Lemma \ref{radem} (and \cite{pauli}), where for a matrix $A\in\mathbbm{C}^{n_1\times n_2}$ the operator $|A)(A|$ maps $X\mapsto A\langle A,X\rangle$ and write
\begin{align}
&\mathcal{X} \coloneqq \sup_{X\in \mathcal{T}}\langle(\tilde{\mathcal{A}}^{*}\tilde{\mathcal{A}} - \mathcal{I})(X),X\rangle \coloneqq \|\tilde{\mathcal{A}}^{*}\tilde{\mathcal{A}} - \mathcal{I}\|_{\mathcal{T}} \nonumber \\ = &\Bigg\|\sum_{k=1}^{m}\Big(|\tilde{\mathcal{A}}^{k})(\tilde{\mathcal{A}}^{k}| - \frac{1}{m}\mathcal{I}\Big)\Bigg\|_{\mathcal{T}} \nonumber \\ = &\Bigg\|\sum_{k=1}^{m}\Big(|\tilde{\mathcal{A}}^{k})(\tilde{\mathcal{A}}^{k}| - \mathbbm{E}|\tilde{\mathcal{A}}^{k})(\tilde{\mathcal{A}}^{k}|\Big)\Bigg\|_{\mathcal{T}}, \nonumber
\end{align}
where the last equality holds due to isotropy of the ensemble (\ref{isotropy}) with i.i.d. samples. $\mathbbm{E}\mathcal{X}$ will be bounded first, and then a concentration inequality will be applied to show this random variable is concentrated around its mean. 

Using symmetrization (as in equation (42) of \cite{pauli}, which uses Lemma 6.3 by \cite{prob}, gives
\[
\mathbbm{E}\mathcal{X} \leq 2\mathbbm{E}_{\tilde{\Omega}}\mathbbm{E}_{\epsilon}\Bigg\|\sum_{k=1}^{m}\epsilon_{k}|\tilde{\mathcal{A}}^{k})(\tilde{\mathcal{A}}^{k}|\Bigg\|_{\mathcal{T}},
\]
where $\epsilon_{k}$ are Rademacher random variables. Applying Lemma \ref{radem}, which requires $m\leq n_1n_2$, gives
\[
\mathbbm{E}_{\epsilon}\Bigg\|\sum_{k}\epsilon_{k}|\tilde{\mathcal{A}}^{k})(\tilde{\mathcal{A}}^{k}|\Bigg\|_{\mathcal{T}} \leq C_1\Bigg\|\sum_{k}|\tilde{\mathcal{A}}^{k})(\tilde{\mathcal{A}}^{k}|\Bigg\|_{\mathcal{T}}^{1/2},
\]
where
\[
C_1 \coloneqq \frac{\tilde{C}\sqrt{M_1\rho}\log(n_1+n_2)\log^{1/2}(m)}{\sqrt{m}},
\]
and $\tilde{C}>0$ is an absolute constant given in Lemma \ref{radem}. Summarizing and continuing these calculations,

\begin{align}
&\mathbbm{E}\mathcal{X} \leq 2C_1\mathbbm{E}\left(\Bigg\|\sum_{k}|\tilde{\mathcal{A}}^{k})(\tilde{\mathcal{A}}^{k}|\Bigg\|_{\mathcal{T}}\right)^{1/2} \nonumber \\
\leq &2C_1\mathbbm{E}\left(\Bigg\|\sum_{k}\Big(|\tilde{\mathcal{A}}^{k})(\tilde{\mathcal{A}}^{k}| - \frac{1}{m}\mathcal{I}\Big)\Bigg\|_{\mathcal{T}}+1\right)^{1/2} \nonumber \\
\leq &2C_1\left(\mathbbm{E}\Bigg\|\sum_{k}\Big(|\tilde{\mathcal{A}}^{k})(\tilde{\mathcal{A}}^{k}| - \frac{1}{m}\mathcal{I}\Big)\Bigg\|_{\mathcal{T}}+1\right)^{1/2} \nonumber \\
= &2C_1\left(\mathbbm{E}\mathcal{X} + 1\right)^{1/2}. \nonumber 
\end{align}

Therefore
\[
\frac{\mathbbm{E}\mathcal{X}}{\sqrt{\mathbbm{E}\mathcal{X} + 1}} \leq \frac{2\tilde{C}\sqrt{M_1\rho}\log(n_1+n_2)\log^{1/2}(m)}{\sqrt{m}} \leq \frac{2\sqrt{2}\tilde{C}\sqrt{M_1\rho}\log^{3/2}(n_1+n_2)}{\sqrt{m}},
\]
where the last inequality holds since $m\leq (n_1+n_2)^2$ by assumption. Given $\delta >0$, $\mathbbm{E}\mathcal{X} \leq \delta$ if
\begin{equation}
\label{sampleLR}
\frac{2\sqrt{2}\tilde{C}\sqrt{M_1\rho}\log^{3/2}(n_1+n_2)}{\sqrt{m}} \leq \frac{\delta}{\sqrt{\delta + 1}}.
\end{equation}

A concentration inequality will show that $\mathcal{X}$ is close to its expected value with high probability.

\begin{theorem}[Theorem 8.42 in (\cite{introCS})]
\label{kleinLR}
Let $\mathcal{F}$ be a countable set of functions $f:\mathbbm{C}^{n_1\times n_2}\mapsto\mathbbm{R}$. Let $Y_1, . . . , Y_m$ be independent random matrices in $\mathbbm{C}^{n_1\times n_2}$ such that $\mathbbm{E}f(Y_k) = 0$ and $f(Y_k) \leq K$ almost surely for all $k\in [m]$ and for all $f \in \mathcal{F}$. Define $Z$ as the random variable
\[
Z = \sup_{f \in F}\sum_{k = 1}^mf(Y_{k}).
\]
Let $\sigma^2>0$ be such that $\mathbbm{E}\sum_{k=1}^{m}f(Y_{\ell})^2 \leq \sigma^2$ for all $f\in\mathcal{F}$. Then for all $\delta\geq 0$
\[
\mathbbm{P}\left(Z\geq \mathbbm{E}Z + \delta \right) \leq \exp\left(-\frac{\delta^2}{2\sigma^2 + 4K\mathbbm{E}Z + 2\delta K/3}\right).
\]
\end{theorem}

To apply the theorem, let $X\in \mathcal{T}$ generate a set of functions $f_X:\mathbbm{C}^{n_1\times n_2}\to \mathbbm{R}$ via
\[
f_X(Z) \coloneqq |\langle Z,X\rangle|^2 - \frac{1}{m}.
\]
Then notice that
\begin{align}
&\mathcal{X} \coloneqq \Bigg\|\sum_{k}\Big(|\tilde{\mathcal{A}}^{k})(\tilde{\mathcal{A}}^{k}| - \frac{1}{m}\mathcal{I}\Big)\Bigg\|_{\mathcal{T}} \nonumber \\
\coloneqq&\sup_{X\in\mathcal{T}}\sum_{k}\Big\langle\Big(|\tilde{\mathcal{A}}^{k})(\tilde{\mathcal{A}}^{k}| - \frac{1}{m}\mathcal{I}\Big)(X),X\Big\rangle \nonumber \\
=&\sup_{X\in\mathcal{T}}\sum_{k}f_X\Big(\tilde{\mathcal{A}}^{k}\Big) = \sup_{X\in\tilde{\mathcal{T}}}\sum_{k}f_X\Big(\tilde{\mathcal{A}}^{k}\Big), \nonumber
\end{align}
where $\tilde{\mathcal{T}}$ is a dense countable subset of $\mathcal{T}$. 

For all $k\in [m]$ and $X\in\mathcal{T}$, by the first part of Lemma \ref{lemma1LR2}
\[
f_X\Big(\tilde{\mathcal{A}}^{k}\Big) \leq |\langle \tilde{\mathcal{A}}^{k},X\rangle|^2 \leq \frac{M_1 \rho}{m},
\]
and $K = \frac{M_1 \rho}{m}$ from Theorem \ref{kleinLR}.

Now for $\sigma^2$, apply the second part of Lemma \ref{lemma1LR2} and isometry of the ensemble to obtain
\[
\mathbbm{E}\sum_{k}f_X\Big(\tilde{\mathcal{A}}^{k}\Big)^2 \leq \mathbbm{E}\sum_{k}|\langle \tilde{\mathcal{A}}^{k},X\rangle|^4 \leq \frac{M_1 \rho}{m}  \coloneqq \sigma^2.
\]

To finish, assuming
\begin{equation}
\label{eqmn}
\sqrt{m} \geq 2\sqrt{2} \tilde{C}\sqrt{M_1\rho}\log^{3/2}(n_1+n_2)\frac{\sqrt{1+\delta}}{\delta},
\end{equation}
gives $\mathbbm{E}\mathcal{X}\leq \delta$ according to (\ref{sampleLR}) and by Theorem \ref{kleinLR}
\[
\mathcal{X} \leq \mathbbm{E}\mathcal{X} + \delta < 2\delta
\]
with probability of failure not exceeding
\[
\exp\left(-\frac{m\delta^2}{2M_1 \rho + 4M_1 \rho\delta + 2M_1 \rho\delta/3}\right) \leq \exp\left(-\frac{6m\delta^2}{19M_1 \rho}\right).
\]
The last inequality holds assuming $\delta \leq \frac{1}{4}$, under which (\ref{eqmn}) holds if
\begin{equation}
\sqrt{m} \geq \frac{\sqrt{10} \tilde{C}\sqrt{M_1\rho}\log^{3/2}(n_1+n_2)}{\delta}, \nonumber
\end{equation}
where the statement of the theorem absorbs all the absolute constants into $C$.

\end{proof}

\subsection*{C.3 Proof of Additional Lemmas}
\label{extralemmas}

This section supplies the proofs of Lemmas \ref{opnorm}, \ref{lemma1LR2} and \ref{radem}.

\begin{proof}[Proof of Lemma \ref{opnorm}]
The main ingredient is a matrix Bernstein inequality (Theorem 1.6 by \cite{userfriend}). As in Section \href{proofconcentration}{C.2}, expand
\[
\tilde{\mathcal{A}}^{*}\tilde{\mathcal{A}}(Z)-Z = \sum_{k=1}^{m}\left(\tilde{\mathcal{A}}^k\langle \tilde{\mathcal{A}}^{k},Z\rangle-\frac{1}{m}Z\right)
\]
which is a sum of independent and centered random matrices. In order to apply the Bernstein inequality, the operator norms of each summand and the matrix variance statistic need to be bounded.

Notice that for any $k\in [m]$
\[
\Big\|\tilde{\mathcal{A}}^k\langle \tilde{\mathcal{A}}^{k},Z\rangle-\frac{1}{m}Z\Big\| \leq \frac{n_1n_2}{m}\|Z\|_{\infty}\coloneqq R
\]
and 
\begin{align}
&\mathbbm{E}\sum_k\left(\tilde{\mathcal{A}}^k\langle \tilde{\mathcal{A}}^{k},Z\rangle-\frac{1}{m}Z\right)\left(\tilde{\mathcal{A}}^k\langle \tilde{\mathcal{A}}^{k},Z\rangle-\frac{1}{m}Z\right)^{\top} = \sum_k\left(\mathbbm{E}\tilde{\mathcal{A}}^k\tilde{\mathcal{A}}^{k\top}|\langle \tilde{\mathcal{A}}^{k},Z\rangle|^2-\frac{1}{m^2}ZZ^{\top}\right) \nonumber \\
&= \frac{n_1n_2}{m}M_1 - \frac{1}{m}ZZ^{\top}, \nonumber
\end{align}
where $M_1\in\mathbbm{C}^{n_1\times n_1}$ is a diagonal matrix whose entries are the diagonal elements of $ZZ^{\top}$. Therefore,
\begin{align}
&\Big\|\mathbbm{E}\sum_k\left(\tilde{\mathcal{A}}^k\langle \tilde{\mathcal{A}}^{k},Z\rangle-\frac{1}{m}Z\right)\left(\tilde{\mathcal{A}}^k\langle \tilde{\mathcal{A}}^{k},Z\rangle-\frac{1}{m}Z\right)^{\top}\Big\| \leq \frac{1}{m}\left(n_1n_2\|M_1\|+\|ZZ^{\top}\|\right) \nonumber \\
&= \frac{1}{m}\left(n_1n_2\max_{1\leq k\leq n_1}\|Z_{k*}\|_2^2 + \|ZZ^{\top}\|\right) \leq \frac{2n_1n_2}{m}\|Z\|_{\infty,2}^2 \coloneqq \sigma^2. \nonumber
\end{align}
The last inequality holds by Gershgorin circle theorem, which gives that for some $k\in [n_1]$
\[
\|ZZ^{\top}\| \leq |(ZZ^{\top})_{kk}| + \sum_{\ell\neq k}|(ZZ^{\top})_{k\ell}| = \|Z_{k*}\|_2^2 + \sum_{\ell\neq k}|\langle Z_{k*},Z_{\ell*}\rangle| \leq n_1\|Z\|_{\infty,2}^2.
\]

Analogously, bound
\[
\Big\|\mathbbm{E}\sum_k\left(\tilde{\mathcal{A}}^k\langle \tilde{\mathcal{A}}^{k},Z\rangle-\frac{1}{m}Z\right)^{\top}\left(\tilde{\mathcal{A}}^k\langle \tilde{\mathcal{A}}^{k},Z\rangle-\frac{1}{m}Z\right)\Big\| \leq \sigma^2.
\] 
Apply Theorem 1.6 in by \cite{userfriend} with $R,\sigma^2$ above and
\begin{align}
&t = \frac{4n_1n_2\log(n_1+n_2)\|Z\|_{\infty}/3}{2m} + \nonumber \\
&\frac{\sqrt{16n_1^2n_2^2\log^2(n_1+n_2)\|Z\|_{\infty}^2/9+32mn_1n_2\log(n_1+n_2)\|Z\|_{\infty,2}^2}}{2m} \nonumber
\end{align}
to obtain the desired probability of success. The statement of the lemma simplifies the upper bound on the operator norm by noting that
\[
t \leq \frac{4n_1n_2\log(n_1+n_2)\|Z\|_{\infty}/3+2\sqrt{2mn_1n_2\log(n_1+n_2)}\|Z\|_{\infty,2}}{m}.
\]

\end{proof}

Lemma \ref{lemma1LR2}, which admits a straightforward proof.

\begin{proof}[Proof of Lemma \ref{lemma1LR2}]
The inequality (\ref{lemma1LReq2}) is straightforward by definition and scaling. For the remaining claim, if $Z\in T$ then by (\ref{lemma1LReq2})
\begin{align}
&\mathbbm{E}\sum_{k=1}^{m}|\langle\tilde{\mathcal{A}}^k,Z\rangle|^4 = \frac{n_1n_2}{m}\sum_{p=1}^{n_1}\sum_{q=1}^{n_2}|Z_{pq}|^4 \nonumber \\
&\leq \frac{n_1n_2}{m}\left(\max_{(p,q)\in[n_1]\times [n_2]}|Z_{pq}|^2\right)\sum_{p=1}^{n_1}\sum_{q=1}^{n_2}|Z_{pq}|^2 \leq \|Z\|_F^4\frac{M_1 \rho}{m}. \nonumber
\end{align}

\end{proof}

The proof of Lemma \ref{radem} is due to \cite{pauli}, tailored here to fit the specific setting. This modified argument results in a tighter bound in terms of the logarithmic dependency. Adopting the author's notation, in what follows for a matrix $A\in\mathbbm{C}^{n_1\times n_1}$ denote $|A)(A|$ as the operator that maps $X\mapsto A\langle A,X\rangle$.

\begin{proof}[Proof of Lemma \ref{radem}]
The argument will rely on the work of \cite{pauli} for brevity, referring the reader to the proof of Lemma 3.1 in Section A therein. With $U_2$ defined by \cite{pauli}, notice that $T\cap S \subset U_2$ since every matrix in $T$ is rank $\rho$ (where $\rho$ is defined in (\ref{dimT})). The result here is obtained in a similar manner, but considering non-square matrices and linear subspace $T$ as Banach space in its own right to compute its covering number. To this end, it is important notice that $T$ equipped with the Frobenius norm is a Banach space and $\{\tilde{\mathcal{A}}^k\}_{k=1}^m \subset T^{*}$, where $T^{*}$ denotes the dual space of $T$, have dual norm bounded as
\begin{equation}
\label{dualnorm}
\|\tilde{\mathcal{A}}^k\|_{T^{*}} = \sup_{X\in T\cap S}\left|\langle \tilde{\mathcal{A}}^k,X\rangle\right| \leq \sqrt{\frac{M_1\rho}{m}} \coloneqq \sqrt{\rho}K, \ \ \ \ \forall k\in[m]
\end{equation}
by Lemma \ref{lemma1LR2}. Notice that $K \coloneqq \sqrt{M_1/m}$.

With this in mind, proceed as in the proof of Lemma 3.1 by \cite{pauli} (with $T\cap S$ replacing $U_2$) up to equation (15) which is obtained via comparison principle to a Gaussian process and Dudley's inequality. Combined with bound (18) therein and a change of variables shows
\[
\mathbbm{E}_{\epsilon}\sup_{X\in T\cap S} \sum_{k=1}^{m}\epsilon_k \langle|\tilde{\mathcal{A}}^k)(\tilde{\mathcal{A}}^k|(X),X\rangle \leq 48\sqrt{2\pi}R\sqrt{\rho}\int_{0}^{\infty}\log^{1/2}\mathcal{N}\left(\frac{1}{\sqrt{\rho}}\left(T\cap S\right),\|\cdot\|_X,\epsilon\right)d\epsilon,
\]
where $\mathcal{N}\left(B,\|\cdot\|,\epsilon\right)$ is the number of balls of radius $\epsilon$ in a metric $\|\cdot\|$ needed to cover a set $B$, 
\[
R \coloneqq \left(\sup_{X\in T\cap S} \sum_{k=1}^{m} \langle|\tilde{\mathcal{A}}^k)(\tilde{\mathcal{A}}^k|(X),X\rangle\right)^{1/2} 
\]
and $\|\cdot\|_X$ is a semi-norm defined for $Y\in \mathbbm{C}^{n_1\times n_2}$ as
\[
\|Y\|_X = \max_{k\in [m]}|\langle \tilde{\mathcal{A}}^k,Y\rangle|.
\]
To bound the integral, bound $\mathcal{N}\left(\frac{1}{\sqrt{\rho}}\left(T\cap S\right),\|\cdot\|_X,\epsilon\right)$ in two different ways. For $Y\in \frac{1}{\sqrt{\rho}}\left(T\cap S\right)$, notice that $\|Y\|_X \leq K$ by (\ref{dualnorm}) so that
\begin{equation}
\label{bou0}
\mathcal{N}\left(\frac{1}{\sqrt{\rho}}\left(T\cap S\right),\|\cdot\|_X,\epsilon\right) \leq \mathcal{N}\left(K\cdot B_X,\|\cdot\|_X,\epsilon\right),
\end{equation}
where $B_X$ is the unit ball in $\|\cdot\|_X$. For small $\epsilon$ and with $n_1\geq n_2$, use (\ref{bou0}) and equation (20) by \cite{pauli} to obtain
\begin{equation}
\label{bou1}
\mathcal{N}\left(\frac{1}{\sqrt{\rho}}\left(T\cap S\right),\|\cdot\|_X,\epsilon\right) \leq \left(1+\frac{2K}{\epsilon}\right)^{2n_1^2}.
\end{equation}
For large $\epsilon$, apply Lemma 3.2 by \cite{pauli} (Lemma 1 by \cite{lemma1}) with $E = T$ equipped with the Frobenious norm to obtain
\begin{equation}
\label{covnumpauli}
\mathcal{N}\left(\frac{1}{\sqrt{\rho}}\left(T\cap S\right),\|\cdot\|_X,\epsilon\right) = \mathcal{N}\left(T\cap S,\|\cdot\|_X,\epsilon\sqrt{\rho}\right) \leq \exp\left(\frac{C_1^2K^2 \log(m)}{\epsilon^2}\right),
\end{equation}
where $C_1$ is an absolute constant given by Maurey's empirical method. The inequality holds by (\ref{dualnorm}) and since $T$ has modulus of convexity of power type 2 with constant $\lambda(T) = 1$ and dual space type 2 constant $T_2(T^{*}) \leq 1$ due to the Frobenius norm (see Theorem A3 and A4 in the Appendix of \cite{pauli1}).

To bound the integral, split it at $A \coloneqq K/n_1$ and use (\ref{bou1}) for small $\epsilon$ to obtain
\begin{align}
&\int_{0}^{A}\log^{1/2}\mathcal{N}\left(\frac{1}{\sqrt{\rho}}\left(T\cap S\right),\|\cdot\|_X,\epsilon\right)d\epsilon \leq \int_{0}^{A}\sqrt{2}n_1\log^{1/2}\left(1+\frac{2K}{\epsilon}\right)d\epsilon \nonumber \\
&\leq \sqrt{2}n_1\int_{0}^{A}\left(1 + \log\left(1+\frac{2K}{\epsilon}\right)\right)d\epsilon = \sqrt{2}n_1\int_{1/A}^{\infty}\left(1 + \log\left(1+2Ky\right)\right)\frac{dy}{y^2} \nonumber \\
&\leq \sqrt{2}n_1\int_{1/A}^{\infty}\left(1 + \log\left((A+2K)y\right)\right)\frac{dy}{y^2} \leq \sqrt{2}K(2+\log(1+2n_1)), \nonumber 
\end{align}
where the final bound holds by integrating $\log\left((A+2K)y\right)/y^2$ by parts.

Consider the remaining part of the integral up to $K$, since when $\epsilon>K$ by (\ref{bou0}) it holds that $\mathcal{N}\left(\frac{1}{\sqrt{\rho}}\left(T\cap S\right),\|\cdot\|_X,\epsilon\right) = 1$. Using (\ref{covnumpauli}) gives
\begin{align}
&\int_{A}^{K}\log^{1/2}\mathcal{N}\left(\frac{1}{\sqrt{\rho}}\left(T\cap S\right),\|\cdot\|_X,\epsilon\right)d\epsilon \leq \int_{A}^{K}\frac{C_1K \log^{1/2}(m)}{\epsilon}d\epsilon = C_1K \log^{1/2}(m)\log(n_1). \nonumber
\end{align}

In conclusion
\begin{align}
&\mathbbm{E}_{\epsilon}\sup_{X\in T\cap S} \sum_{k=1}^{m}\epsilon_k \langle|\tilde{\mathcal{A}}^k)(\tilde{\mathcal{A}}^k|(X),X\rangle \nonumber \\
&\leq 48\sqrt{2\pi}R\sqrt{\rho}\left(\sqrt{2}K(2+\log(1+2n_1))+C_1K \log^{1/2}(m)\log(n_1)\right) \nonumber \\
&\leq \frac{\tilde{C}\sqrt{M_1\rho}\log^{1/2}(m)\log(n_1+n_2)}{\sqrt{m}}\left(\sup_{X\in T\cap S} \sum_{k=1}^{m} \langle|\tilde{\mathcal{A}}^k)(\tilde{\mathcal{A}}^k|(X),X\rangle\right)^{1/2}, \nonumber
\end{align}
for some absolute constant $\tilde{C}$.

\end{proof}

The main difference in this proof as opposed to the proof of Lemma 3.1 by \cite{pauli} is that the containment of $\frac{1}{\sqrt{\rho}}\left(T\cap S\right)$ in the nuclear norm ball $B_1$ is not considered here. Rather, the linear subspace $T$ is viewed as a Banach space with unit ball $T\cap S$. This allows for a direct application of Lemma A.3, reducing the logarithmic dependency by a factor of $\log^{3/2}(n_1)$. The same gain does not seem straightforward for the universal result of \cite{pauli}. Otherwise, using the arguments here, the sampling complexity of Pauli measurements can be reduced to $\mathcal{O}(n_1\rho\log^{3/2}(n_1))$ when universality is not imposed (for example, via a dual certificate instead of the restricted isometry property).

\vskip 0.2in
\bibliography{main}

\end{document}